\documentclass[11pt]{article}
\newcommand{\ignore}[1]{}

\usepackage[normalem]{ulem}

\usepackage{amsfonts,amsmath,amsthm,times,fullpage,amssymb,epsfig,graphicx,enumerate,bm}

\usepackage{hyperref}

\hypersetup{colorlinks=false}

\usepackage[dvipsnames]{xcolor} 

\usepackage{color}
\usepackage{dsfont}
\usepackage{wrapfig}
\usepackage{relsize}

\usepackage{url} 

\usepackage{algorithm} 
\usepackage{algorithmic}

\usepackage[left=1.0in,top=1.0in,right=1.0in,bottom=1.0in,nohead,textheight=10in,footskip=0.3in]{geometry}

\usepackage{caption}

\newtheorem{theorem}{Theorem}[section]
\newtheorem{lemma}[theorem]{Lemma}
\newtheorem{corollary}[theorem]{Corollary}

\newtheorem{definition}[theorem]{Definition} 
\newtheorem{remark}{Remark}[section]
\newtheorem{proposition}[theorem]{Proposition}

\newcommand\ex{{\mathbb{E}}}

\newcommand{\beq}{\begin{equation}}
\newcommand{\eeq}{\end{equation}}

\newcounter{fooTH}
\newcounter{fooEQ}

\begin{document}

\title{Improved bounds for  coloring  locally sparse hypergraphs}

\author{
Fotis Iliopoulos
\thanks{This material is based upon work directly supported by the IAS Fund for Math and indirectly supported by the National Science Foundation Grant No. CCF-1900460. Any opinions, findings and conclusions or recommendations expressed in this material are those of the author(s) and do not necessarily reflect the views of the National Science Foundation. This work is also supported by the National Science Foundation Grant No. CCF-1815328.} \\ 
Institute for Advanced Study   \\ 
and Princeton University \\
{\small fotios@ias.edu}
}

\date{\empty}

\maketitle

\begin{abstract}

We  show that, for every $k \ge 2$, every $k$-uniform hypergaph of degree $\Delta$ and girth at least $5$ is efficiently $(1+o(1) )(k-1)  (\Delta / \ln \Delta )^{ 1/(k-1) } $-list colorable. As an application (and to the best of our knowledge) we obtain the currently best algorithm for list-coloring random  hypergraphs of bounded average degree.

 \end{abstract}

\thispagestyle{empty}

\newpage

\setcounter{page}{1}

\newpage
\section{Introduction}

In hypergraph coloring one is given a hypergraph $H(V,E)$ and the goal is to find an assignment of one of $q$ colors to each vertex $v \in V$ so that no hyperedge is monochromatic. In the more general \emph{list-coloring} problem, a list of $q$ allowed colors is specified for each vertex.  A graph is $q$-list-colorable if it has a list-coloring no matter how the lists are assigned to each vertex.
The \emph{list chromatic number}, $\chi_{\ell}(H) $, is the smallest $q$ for which $H$ is $q$-list colorable.

Hypergraph coloring is a fundamental constraint satisfaction problem with several applications in computer science and combinatorics, that has been studied for over 60 years. In this paper we consider the task of  coloring locally sparse hypergraphs and its connection to coloring sparse random hypegraphs.

A hypergraph is  \emph{$k$-uniform} if every hyperedge contains exactly $k$ vertices. An $i$-cycle in a $k$-uniform hypergraph is a collection of $i$ distinct hyperedges spanned by at most $i(k-1)$ vertices. We say that a $k$-uniform hypergraph has girth at least $g$ if it contains no $i$-cycles for $2 \le i < g$.  Note that if a $k$-uniform hypergraph has girth at least $3$ then every two of its hyperedges have at most one vertex in common.

The main contribution of this paper is to prove the following theorem.

\begin{theorem}\label{main_hypergraph}
Let $H$ by any $k$-uniform hypergraph, $k \ge 2$, of maximum degree $\Delta$ and girth at least $5$. For all $\epsilon > 0 $,  there exist a positive constant $\Delta_{\epsilon, k} $ such that if $ \Delta \ge \Delta_{\epsilon,k} $, then
\begin{align}\label{our_bound}
\chi_{\ell}(H) \le  (1+\epsilon)(k-1) \left( \frac{ \Delta}{ \ln  \Delta } \right)^{ \frac{1}{ k-1} }.
\end{align}
Furthermore, if $H$ is a hypergraph on $n$ vertices then there exists an algorithm that constructs such a coloring  in expected polynomial time in $n$. 
\end{theorem}
\begin{remark}\label{det_remark}
Theorem~\ref{main_hypergraph} makes no assumption about how $\Delta$ relates to $n$. However,  if $\Delta$ is assumed to be constant, then the algorithm of Theorem~\ref{main_hypergraph} can be efficiently derandomized, i.e., there exists a \emph{deterministic} algorithm that constructs the promised coloring in polynomial time in $n$.
\end{remark}

 Theorem~\ref{main_hypergraph} is interesting for a number of reasons. First, it generalizes a well-known result of Kim~\cite{kim1995brooks} for coloring graphs of degree $\Delta$ and girth $5$, and it implies the classical theorem of Ajtai, Koml\'{o}s, Pintz, Spencer and Szemer\'{e}di~\cite{ajtai1982extremal} regarding the independence number of $k$-uniform hypergraphs of degree $\Delta$ and girth $5$. The latter is a seminal result in combinatorics, with applications in geometry and coding theory~\cite{komlos1982lower,kostochka2001chromatic,lefmann2005sparse}.
 Second, Theorem~\ref{main_hypergraph}  is tight up to a $k$-dependent constant~\cite{bohman2010coloring}. Note also  that, without the girth assumption, the best possible bound~\cite{LLL} on the chromatic number of  $k$-uniform hypergraphs is $O( \Delta^{1/(k-1) }  )$, i.e., it is asymptotically worse than the one of Theorem~\ref{main_hypergraph}. For example, there exist graphs of degree $\Delta$ whose chromatic number is exactly $\Delta+1$.  Third, when it applies, Theorem~\ref{main_hypergraph} improves upon a result of Frieze and Mubayi~\cite{frieze2013coloring} regarding the chromatic number of simple hypergraphs, who showed~\eqref{our_bound} with an unspecified large  leading constant (of order at least $\Omega(k^4)$).     Finally,    Theorem~\ref{main_hypergraph}  can be used to provide the currently best algorithm for list-coloring  random $k$-uniform hypergarphs of bounded average degree (to the best of our knowledge).  We discuss the connection between locally sparse hypergraphs and  sparse random hypergraphs with respect to the task of coloring in the following section.

  \subsection{Application to coloring pseudo-random hypergraphs}

  The random  $k$-uniform hypergraph $H(k,n,p)$ is obtained by choosing each of the  ${ n \choose k } $ $k$-element subsets of a vertex set $V$ ($|V| = n$ ) independently with probability $p$. The chosen subsets are the hyperedges of the hypergraph. 
Note that for $k = 2 $ we have the usual definition of the random graph $G(n,p) $. We say that $H(k,n,p)$  has a certain property $A$ \emph{asymptotically almost surely} or \emph{with high probability}, if the probability that $H \in H(k,n,p) $  has $A$ tends to $1$ as $n \to \infty$.  

In this paper we are  interested in $H(k ,n, d /  {n \choose k-1 }   )$, i.e., the family of random $k$-uniform hypergraphs of bounded average degree $d$. Specifically, we use Theorem~\ref{main_hypergraph}  to prove the following theorem.

 \begin{theorem}\label{randomized_result}
 For any constants $\delta   \in (0,1)$, $k \ge 2$,  there exists $d_{\delta,k} > 0$ such that for every constant $d \ge d_{\delta,k}$,  the random hypergraph $H(k ,n, d /  {n \choose k-1 }   )$ can be   $(1+\delta) (k-1) ( d / \ln d )^{1/(k-1)}$-list-colored by a deterministic algorithm whose running time is polynomial in $n$ asymptotically almost surely. 
 \end{theorem}

\begin{remark}
 Note that, for $k,d$ constants, a very standard argument reveals that $H(k,n,d/{n \choose k-1 })$  is essentially equivalent to $\mathbb{H}(k,n, dn/k )$, namely the uniform distribution over $k$-uniform hypergraphs with $n$ vertices and exactly $dn/k$ hyperedges. Thus, Theorem~\ref{randomized_result} extends to that model as well.
\end{remark}

We note that previous approaches~\cite{achlioptas1997analysis,vu2000choice, frieze2013coloring}  for list-coloring random $k$-uniform hypergraphs of bounded average degree $d$ are either randomized, or require significantly larger lists of colors per vertex in order to succeed.  
Indeed, for $k=2$ the approach of Achlioptas and Molloy~\cite{achlioptas1997analysis} matches the bound of Theorem~\ref{randomized_result} but is randomized, while for $k \ge 3$, and to the best our knowledge, our algorithm uses less colors that any algorithm for  list-coloring random hypergraphs of bounded average degree that has been rigorously analyzed in the literature. Moreover, it is believed that all efficient algorithms  require lists of size at least $(1+o(1) ) ((k-1) d / \ln d)^{1/(k-1)}    $, as this bound corresponds to the so-called \emph{shattering threshold}~\cite{mitsaras_barriers,ayre2019hypergraph,gabrie2017phase}  for coloring sparse random hypergraphs, which is also often referred to as the ``algorithmic barrier"~\cite{mitsaras_barriers}. This threshold arises in a plethora of random constraint satisfaction problems, and it corresponds to a precise phase transition in the geometry set of solutions. In all of these problems, we are not aware of any efficient algorithm that works beyond the algorithmic barrier, despite the fact that solutions exist for constraint-densities larger than the one in which the shattering phenomenon appears.   We refer the reader to~\cite{mitsaras_barriers,zdeborova2007phase} for further details.

In order to prove Theorem~\ref{randomized_result}, we show that random $k$-uniform hypergraphs of bounded average degree $d$ can essentially be treated as hypergraphs of girth $5$ and maximum degree $d$ for the purposes of list-coloring, and then apply Theorem~\ref{main_hypergraph}. In particular, we identify a pseudo-random family of hypergraphs which we call \emph{girth-reducible}, and show that almost all $k$-uniform hypergraphs of bounded average degree belong in this class.  Then we show that girth-reducible hypergraphs can be colored efficiently using Theorem~\ref{main_hypergraph}.

Formally, a $k$-uniform  hypergraph $H$ is \emph{$\kappa$-degenerate} if the induced subhypergraph of all subsets of its vertex set has a vertex of degree at most $\kappa$. The \emph{degeneracy} of a hypergraph  $H$ is the smallest value of $\kappa$ for which $H$ is $\kappa$-degenerate. Note that it is known  that $\kappa$-degenerate hypergraphs are $(\kappa+1)$-list colorable and that the degeneracy of a hypergraph can be computed efficiently by an algorithm that repeatedly removes minimum degree vertices. 
Indeed, to list-color a $\kappa$-degenerate hypergraph we repeatedly find a vertex with (remaining) degree at most $\kappa$, assign to it a color that does not appear in any of its neighbors so far, and remove it from the hypergraph. Clearly, if the lists assigned to each vertex are of size at least $\kappa+1$ this procedure always terminates successfully.

\begin{definition}\label{pseudo_random_property}
For $\delta  \in (0,1)$, we say that a $k$-uniform hypergraph $H(V,E)$ of average degree $d$ is  \emph{$\delta$-girth-reducible}  if its vertex set can be partitioned in two sets, $U$ and $V\setminus U$, such that:
\begin{enumerate}[(a)]

\item  $U$ contains all cycles of length at most $4$, and all vertices of degree larger than $(1+\delta)d$;

\item subhypergraph $H[U] $ is $  \left(\frac{ d  } {\ln d } \right)^{\frac{1}{k-1} } $-degenerate;

\item every vertex in $V \setminus U$ has at most  $\delta  \left( \frac{ d }{ \ln d} \right)^{ \frac{1}{k-1} }$ neighbors in $U$.\label{def_c}

\end{enumerate}
\end{definition}

In words, a hypergraph is $\delta$-girth-reducible if its vertex set can be seen as the union of two parts: A ``low-degeneracy" part, which contains all vertices of degree more than $(1+\delta)d$ and all cycles of lengths at most $4$, and a ``high-girth" part, which induces a hypergraph of maximum degree at most $(1+\delta)d$ and girth $5$. Moreover, each vertex in the ``high-girth" part has only a few neighbors in the ``low-degeneracy" part.

Note   that given a $\delta$-girth-reducible hypergraph we can  efficiently find the promised partition $(U, V\setminus U)$ as follows. We start with $U := U_0$, where $U_0$ is the set of vertices that either have degree at least $(1+\delta) d$, or they are contained in a cycle of length at most $4$. Let $\partial U$ denote the vertices in $V \setminus U$ that violate property~\eqref{def_c}.  While $\partial U \ne \emptyset$, update $U$ as $U: = U \cup \partial U$.
The correctness of the process lies in the fact that in each step we add to the current  $U$ a set of vertices that must be in the low-degeneracy part of the hypergraph.  Observe also that this process allows us to  efficiently check whether a hypergraph is $\delta$-girth-reducible.

 We prove the following theorem regarding the list-chromatic number of girth-reducible hypergraphs.

\begin{theorem}\label{main}
For any constants $\delta \in ( 0,1)$ and $k \ge 2$, there exists $d_{\delta,k} > 0$ such that if $H$ is a $\delta$-girth-reducible,  $k$-uniform hypergraph of average degree $d \ge d_{\delta,k}$, then
\begin{align*}
\chi_{\ell}(H) \le  (1+\epsilon)(k-1) \left( \frac{ d}{ \ln d } \right)^{\frac{1}{k-1} },
\end{align*}
where $\epsilon = 4 \delta = O(\delta)$.  Furthermore, if $H$ is a hypergraph on $n$ vertices then there exists a deterministic algorithm that constructs such a coloring  in time polynomial in $n$.
\end{theorem}
\begin{proof}[Proof of Theorem~\ref{main}]
 Let $\epsilon = 4 \delta$. Given lists of colors of size $(1+\epsilon)(k-1) \left(\frac{d}{\ln d} \right)^{ \frac{1}{k-1} }   $ for each vertex of $H$, we first color the vertices of $U$ using the greedy algorithm which exploits the low degeneracy of $H[U]$. Now each vertex in $V - U$ has at most $\delta \left( \frac{ d }{ \ln d} \right)^{ \frac{1}{k-1} } $ forbidden colors in its list as it has at most that many neighbors in $U$. We delete these colors from the list. Observe that if we manage to properly color the induced subgraph $H[V \setminus U]$ using colors from the updated lists, then we are done since every hyperedge with vertices both in $U$ and $V\setminus U$ will be automatically ``satisfied", i.e., it cannot be monochromatic. Notice now that the updated list of each vertex still contains at least  $  (1 +3 \delta   ) (k-1)\left(\frac{ d}{ \ln d } \right)^{ \frac{1}{k-1} }   $ colors, for sufficiently large $d$. Since the induced subgraph $H[V \setminus U ] $ is of girth at least $5$ and of maximum degree at most $(1+\delta) d$, it is  efficiently $(1+\delta)  (k-1)  \left( \frac{ (1+\delta)d}{ \ln \left((1+\delta) d \right)}\right)^{\frac{1}{k-1}}  $-list-colorable for sufficiently large $d$  per Theorem~\ref{main_hypergraph} and Remark~\ref{det_remark}. This concludes the proof since $(1 + \delta )( 1+ \delta )^{ \frac{1}{k-1} }  < (1 + 3\delta) $.
  
\end{proof}

Moreover, we show that girth-reducibility is a pseudo-random property which is admitted by almost all sparse $k$-uniform hypregraphs.

\begin{theorem}\label{main_property}
 For any constants $\delta   \in (0,1)$, $k \ge 2$,  there exists $d_{\delta,k} > 0$ such that for every constant $d \ge d_{\delta,k}$,  asymptotically almost surely, the random hypergraph $H(k ,n, d /  {n \choose k-1 }   )$ is $\delta$-girth-reducible.
 \end{theorem}

Theorem~\ref{main_property}  follows  by simple, although somewhat technical, considerations on properties of sparse random  hypergraphs, which are mainly inspired by the  results of Alon, Krivelevich and Sudakov~\cite{alon1999list} and  {\L}uczak~\cite{luczak1991chromatic}.  Observe that combining Theorem~\ref{main_property} with Theorem~\ref{main} immediately implies Theorem~\ref{randomized_result}.

Overall, the task of coloring locally sparse hypergraphs is  inherently related to the average-case complexity of coloring. In particular, in this section  we showed that Theorem~\ref{main_hypergraph}  implies a \emph{robust} algorithm for hypergraph coloring, namely a deterministic procedure that applies to worst-case $k$-uniform hypergraphs, while at the same using a number of colors that is only a $(k-1)$-factor away from the algorithmic barrier for random instances (matching it for $k=2$).  We remark that this application is inspired   by  recent results that study the connection between local sparsity and efficient randomized algorithms for coloring sparse regular random graphs~\cite{molloy2019list,AIS,davies2020graph}.

\subsection{Technical overview}\label{technical}

The intuition behind the proof of Theorem~\ref{main_hypergraph} comes from the following   observation, which we explain in terms of graph coloring for simplicity. Let $G$ be a triangle-free graph of  degree $\Delta$, and assume that each of its vertices is assigned an arbitrary  list of $q$ colors. Fix a vertex $v$ of $G$, and consider the random experiment in which  the neighborhood of $v$ is properly list-colored randomly. Since $G$ contains no triangles, this amounts to assigning to each neighbor of $v$ a color from its list randomly and independently. Assuming that  $q \ge q^{*} :=  (1+\epsilon) \Delta/  \ln \Delta   $,  the expected number of \emph{available} colors for $v$, i.e., the colors from the list of $v$  that do not appear in any of its neighbors, is at least $q (1- 1/q  )^{ \Delta } = \omega( \Delta^{ \epsilon/2}) $.  In fact, a simple concentration argument reveals that the number of available colors for $v$ in the end of this experiment is at least $\Delta^{\epsilon/2 }$ with probability that goes to $1$ as $\Delta $ grows. To put it differently, as long as $q \ge q^*$, the vast majority of valid ways to list-color the neighborhood of $v$ ``leaves enough room" to color $v$ without creating any monochromatic edges.

A completely analogous observation regarding the ways to properly color the neighborhood of a vertex can be made for $k$-uniform hypergraphs.  In order to exploit it we employ the so-called  \emph{semi-random method}, which is the main tool behind some of the strongest graph coloring results, e.g.,~\cite{JO,JO2,kahnChrom,kahnListChrom,kang2021proof,molloy1998bound,vu2002general}, including the one of Kim~\cite{kim1995brooks}.  (See also the very recent survey~\cite{kang2021graph} of Kang et. al. on the subject.) The idea is to gradually color the hypergraph in iterations until we reach a point where we can finish the coloring with a simple, e.g., greedy, algorithm. In its most basic form, each iteration consists of the following simple procedure (using graph vertex coloring as a canonical example):  Assign to each vertex a color chosen uniformly at random; then uncolor any vertex that receives the same color as one of its neighbors.  Using the Lov\'{a}sz Local Lemma~\cite{LLL} and concentration inequalities, one typically shows that, with positive probability, the resulting partial coloring has useful properties that allow for the continuation of the argument in the next iteration. (In fact, using the Moser-Tardos algorithm~\cite{MT} this approach yields efficient, and often times deterministic~\cite{MT,determ,harris2023deterministic}, algorithms.) Specifically, one keeps track of certain parameters of the current partial coloring and makes sure that, in each iteration, these parameters evolve almost as if the coloring was totally random. For example, recalling the  heuristic experiment of the previous paragraph, one of the  parameters we would like  to keep track of in our case is a lower bound on  the number of available colors of each vertex in the hypergraph: If this parameter evolves ``randomly" throughout the process, then the  vertices that remain uncolored in the end  are guaranteed to have a non-trivial number of available colors.

Applications of the semi-random method tend to be technically intense and this is even more so  in our case, where we have to deal with constraints of large arity.  Large constraints  introduce several difficulties, but the most important one is that our algorithm  has to control many parameters that interact with each other. Roughly, in order to guarantee the properties that allow for the continuation of the argument in the next iteration, for each uncolored vertex $v$, each color $c$  in the list of $v$, and  each integer $r \in [k-1]$, we should keep track of a lower bound on the  number of adjacent to $v$ hyperedges that have $r$ uncolored vertices and $k-1-r$ vertices colored $c$. Clearly, these parameters are not independent of each other throughout the process, and so the main challenge is to design and analyze a coloring procedure in which all of them, simultaneously,  evolve essentially randomly.

From a technical perspective, our approach generalizes the work of Kim~\cite{kim1995brooks} from graphs to hypergraphs. This distinguishes our work from that of Frieze and Mubayi~\cite{frieze2013coloring}, who instead generalized the approach of Johansson~\cite{JO}. While Johansson’s method—which established that the chromatic number of triangle-free graphs of degree $\Delta$ is $O( \frac{\Delta}{\ln \Delta})$—can handle graphs with $4$-cycles (unlike Kim's), it is known to fall short of the constant corresponding to the algorithmic barrier for coloring sparse random graphs. Consequently, the result of Frieze and Mubayi inherits this limitation. Furthermore, while the recent result of Molloy~\cite{molloy2019list} attains this constant, its technique does not appear to generalize readily to hypergraphs. We demonstrate that Kim's technique, however, does extend to hypergraphs, yielding a bound within a $(k-1)$-factor of the algorithmic barrier—the strongest result to date.

\subsection{Organization of the paper} 
The paper is organized as follows.
In Section~\ref{Background} we present the necessary background.    In Section~\ref{main_hypergaph_proof} we present the algorithm and state the key lemmas for the proof of Theorem~\ref{main_hypergraph}, while in Section~\ref{omitted_hyper} we give the full details.  Finally, in Section~\ref{properties} we prove Theorem~\ref{main_property}.

\section{Background and preliminaries}\label{Background}

In this section we give some background on the technical tools that we will use in our proofs.
\subsection{The  Lov\'{a}sz Local Lemma}

As we have already mentioned, one of the key tools we will use in our proof is the Lov\'{a}sz Local Lemma (LLL)~\cite{LLL}.
\begin{theorem}\label{generalLLL}
Consider a set  $\mathcal{B} = \{B_1, B_2,\ldots,B_m\}$  of (bad) events. For each $B \in \mathcal{B}$, let $D(B) \subseteq	\mathcal{B} \setminus \{B\}$ be such that $\Pr[ B \mid \bigcap_{C \in S} \overline{C}] =  \Pr[B]$ for every $S \subseteq \mathcal{B} \setminus (D(B) \cup \{B\})$. If there is a function $x:\mathcal{B} \rightarrow  (0,1)$ satisfying
\begin{equation}\label{eq:LLL}
\Pr[B] \le x(B) \prod_{C \in D(B)} (1-x(C))  \enspace \text{  for all $B \in \mathcal{B}$}, 
\end{equation}
then the probability that none of the events in $\mathcal{B}$ occurs is at least $\prod_{B \in \mathcal{B}} (1- x(B))> 0$. 
\end{theorem}

In particular, we will need the following two corollaries of Theorem~\ref{generalLLL}. For their proofs, the reader is referred to Chapter 19 in~\cite{mike_book} (see also Remark~\ref{llls}).

\begin{corollary}\label{MikeLLL}
Consider a set  $\mathcal{B} = \{B_1,\ldots,B_m\}$  of (bad) events. For each $B \in \mathcal{B}$, let $D(B) \subseteq	\mathcal{B} \setminus \{B\}$ be such that $\Pr[ B \mid \bigcap_{C \in S} \overline{C}] =   \Pr[B]$ for every $S \subseteq \mathcal{B} \setminus (D(B) \cup \{B\})$.  If for every $B \in \mathcal{B}$:
\begin{enumerate}[(a)]
\item   $\Pr[B] \le  \frac{1}{4} $;

\item$ \sum_{ C \in D(B)  } \Pr[C] \le \frac{1}{4} $, 
\end{enumerate}
 then the probability that none of the events in $\mathcal{B}$  occurs is strictly positive.
\end{corollary}

\begin{corollary}\label{symmetric_LLL}
Consider a set $\mathcal{B} = \{ B_1, B_2, \ldots, B_m \}$  of  (bad) events such that for each $B \in \mathcal{B}$:
\begin{enumerate}[(a)]
\item $ \Pr[ B ] \le p < 1 $; 
\item $B$ is mutually independent of a set of all but at most $\Delta$ of the other events. 
\end{enumerate}
If $4 p \Delta \le 1$ then with positive probability, none of the events in $\mathcal{B}$ occur.
\end{corollary}

\begin{remark}\label{llls}
Corollary~\ref{MikeLLL} follows from Theorem~\ref{generalLLL} by setting $x(B) = 2 \Pr[B]$, while Corollary~\ref{symmetric_LLL} is obtained by setting $x(B) = \frac{1}{\Delta+1}$ for all $B \in \mathcal{B}$.
\end{remark}

\subsection{The Variable Setting and the Moser-Tardos Algorithm}

While the Lov\'{a}sz Local Lemma is a powerful tool for proving the existence of combinatorial structures, the original proof was non-constructive. To address this, we utilize the algorithmic framework developed by Moser and Tardos~\cite{MT}.

We consider the \emph{variable setting}, where the probability space is determined by a collection  of $n$ mutually independent random variables  $\mathcal{V} = \{V_1, \ldots, V_n\}$. Each event $B \in \mathcal{B}$ depends only on a subset of these variables, denoted by $\text{vbl}(B) \subseteq \mathcal{V}$. In this setting, the dependency neighborhood $D(B)$ for an event $B$ is naturally defined as the set of all other events that share at least one variable with $B$. Formally, $D(B) = \{ C \in \mathcal{B} \setminus \{B\} : \text{vbl}(B) \cap \text{vbl}(C) \neq \emptyset \}$.

\subsubsection{The Randomized Algorithm}

The randomized algorithm proposed by Moser and Tardos is remarkably simple: we begin with a random assignment to the variables in $\mathcal{V}$. If any bad event $B \in \mathcal{B}$ occurs, we ``resample'' it by choosing new values for the variables in $\text{vbl}(B)$ according to their underlying distribution. This process is repeated until no bad events occur.

\begin{algorithm}[H]
\caption{Moser-Tardos Resampling Algorithm}
\begin{algorithmic}[1]
\STATE For each $V \in \mathcal{V}$, initialize $V$ with a random value independently.
\WHILE{there exists some $B \in \mathcal{B}$ such that $B$ occurs}
    \STATE Pick an arbitrary occurring event $B \in \mathcal{B}$.
    \STATE Resample all variables $V \in \text{vbl}(B)$ independently.
\ENDWHILE
\RETURN the current assignment of values to $\mathcal{V}$.
\end{algorithmic}
\end{algorithm}

Moser and Tardos proved that if the LLL condition holds, this algorithm finds a valid assignment quickly.

\begin{theorem}[\cite{MT}]\label{thm:randomized_mt}
Let $\mathcal{B}$ be a set of events defined over a collection of independent variables $\mathcal{V}$, as per the variable setting. If there exists a function $x: \mathcal{B} \rightarrow (0,1)$ such that condition \eqref{eq:LLL} holds, then there exists an assignment of values to the variables $\mathcal{V}$ not violating any of the events in $\mathcal{B}$.
Furthermore, the Moser-Tardos algorithm resamples an event $B \in \mathcal{B}$ at most an  expected $\frac{x(B)}{1-x(B)}$ times before it finds such an evaluation. Thus the expected total number of resampling steps is at most $\sum_{B \in \mathcal{B} } \frac{x(B) }{1 - x(B) } $.
\end{theorem}

\subsubsection{Deterministic Algorithms}

Moser and Tardos~\cite{MT} also developed a deterministic version of their resampling algorithm, derived via the method of conditional expectations. This version requires slightly stronger computational assumptions: specifically, we must be able to efficiently compute conditional probabilities of the bad events.

\begin{theorem}[\cite{MT}]\label{thm:deterministic_mt}
Let $\mathcal{B} = \{B_1, B_2, \ldots, B_m \}$ be a set of events defined over a collection of independent variables $\mathcal{V} = \{V_1, V_2, \ldots, V_n \}$, as per the variable setting. Suppose each variable $V_i$ takes values from a finite domain $D_i$. Consider the problem size to be $s:= m + n + \sum_{i=1}^{ n} |D_i|$. 
Suppose there exists an algorithm that can compute, for each $B \in \mathcal{B}$ and each partial evaluation $(v_i \in D_i)_{i \in I}$, $I \subseteq [n]$ the conditional probability $\Pr[B \mid \forall i \in I : V_i = v_i]$ in time polynomial in $s$. Suppose, moreover,  that the size of the dependency neighborhood of each event is bounded by a constant, that is, $\forall B \in \mathcal{B} : |D(B)| \le \lambda $ for some constant $\lambda$. If there is a constant $\varepsilon > 0$ and an assignment of reals $x : \mathcal{B} \rightarrow (0,1)$ such that
\begin{equation}
\forall B \in \mathcal{B} : \Pr[B] \le (1 - \varepsilon)x(B) \prod_{C \in D(B) }(1 - x(C)),
\end{equation}
then a deterministic algorithm can find an evaluation of the variables such that no event occurs in time polynomial in $s$.

\end{theorem}

We also employ a more recent result from~\cite{harris2023deterministic}, which achieves derandomization without requiring the efficient computation of conditional probabilities. While this approach necessitates a somewhat stronger criterion (see Theorem~\ref{harris_theorem_derandomize}), it allows for a ``black-box'' treatment of the underlying variables $\mathcal{V}$.

\begin{theorem}[\cite{harris2023deterministic}]\label{harris_theorem_derandomize}
Let $\mathcal{B} = \{B_1, B_2, \ldots, B_m \}$ be a set of events defined over a collection of independent variables $\mathcal{V} = \{V_1, V_2, \ldots, V_n \}$, as per the variable setting.  Let $d = \max_{B \in \mathcal{B}}| D(B) \cup \{B \} |$ and $p = \max_{B \in \mathcal{B}} \Pr[B]$. If there exists a constant $\epsilon > 0$ such that $\mathrm{e} p d^{1+ \epsilon} \le 1$, then an assignment of values to the variables $\mathcal{V}$ avoiding all events in $\mathcal{B}$ can be found deterministically in time polynomial in $n$ and $m$.
\end{theorem}

\subsection{Talagrand's inequality}

We will also need the following  version of Talagrand's inequality~\cite{talagrand1995concentration} whose proof can be found  in Chapter 20 of~\cite{mike_book}. 
\begin{theorem} \label{talagrand_inequality}
Let $X$ be a non-negative random variable, not identically $0$, which is determined by $n$ independent trials $T_1, \ldots, T_n$, and satisfying the following for some $\gamma > 0$:
\begin{enumerate}
\item changing the outcome of any trial can affect $X$ by at most $\gamma$, and

\item for any $s$, if $X \ge s$ then there is a set of at most $ws$ trials whose outcomes certify that $X \ge s$,
\end{enumerate}
then for any $0 \le t \le \ex[X] $,
\begin{align*}
\Pr[ | X - \ex[X] | > t + 60 \gamma \sqrt{ w\ex[X]  }  ] \le 4 \mathrm{e}^{ - \frac{t^2 }{ 8 \gamma^2 w \ex[X]} }  \enspace.
\end{align*}
\end{theorem}

\section{List-coloring high-girth hypergraphs}\label{main_hypergaph_proof} 
In this  section we describe the algorithm of Theorem~\ref{main_hypergraph}. As we already explained, our approach is based on the semi-random method. For an excellent exposition both of the method and Kim's result the reader is referred to~\cite{mike_book}.

We assume without loss of generality that $\epsilon < \frac{1}{10}$. Also, it will be convenient to define the parameter $\delta:= (1+\epsilon)(k-1) - 1 $, so that the list of each vertex initially has at least $(1+\delta) (\frac{\Delta }{ \ln \Delta })^{\frac{1}{k-1}} $ colors, and assume that $k \ge 3$. (The case $k=2$ is Kim's result.)
Finally, in the statements of all lemmas and corollaries of this section we tacitly assume that $\Delta$ is a sufficiently large constant.

%

The coloring process begins with all vertices of $H$ uncolored and proceeds through a sequence of iterations. At each iteration $i$, we focus on the set of remaining uncolored vertices $V_i$, each having a  \emph{candidate list} of potential colors $L_v = L_v(i)$.  We define a \emph{state} $\sigma$ as an assignment that pairs each $v \in V_i$ with both a color from $L_v$ and an activation bit (designating the vertex as either \emph{activated} or \emph{deactivated}). During the $i$-th iteration, we select a specific state $\sigma_i$ from a probability distribution. This selection is handled using the Lovász Local Lemma to ensure that the chosen state exhibits specific desired properties.  Following this selection, the properties of $\sigma_i$ determine which vertices permanently \emph{retain} their assigned colors; subsequently, the candidate lists $L_v$ are updated, and any vertices that remain uncolored proceed to iteration $i+1$.

The above establishes the framework for the iterative coloring procedure. What remains is to formally specify the probability distribution for each iteration, the sampling mechanism, the criteria for color retention, and the transition rules for updating the candidate lists $L_v$. To this end, it is necessary to first introduce the following definitions and remarks.

We define a color $c \in L_v$ as \emph{available} for vertex $v$ under a state $\sigma$ if the assignment of $c$ to $v$ preserves the property that no hyperedge that contains $v$ is monochromatic within the set consisting of vertices activated under $\sigma$ and those assigned colors in preceding iterations. 

For each vertex $v$, color $c \in L_v$ and iteration $i$, we  define a few quantities of interest that our process will attempt to control. Let $\ell_i(v)$  be the size of $L_v$. Further, for each  $r \in [k]$, let 
$D_{i,r}(v,c)$ denote the set of hyperedges  $h$ that contain $v$ and, in the beginning of the $i$-th iteration: (i) exactly $r$ vertices $\{u_1, \ldots, u_{r} \} \subseteq h \setminus \{v \}$    are uncolored  and $c \in L_{u_j}$ for every $j \in [r]$; (ii) the rest $k-1-r$ vertices of $h$ other than $v$ are colored $c$. We define $t_{i,r}(v,c) := | D_{i,r}(v,c) |$.

As it is common in the applications of the semi-random method, we will not attempt to keep track  of the values of $\ell_i(v)$ and $t_{i,r}(v,c)$, $r \in [k-1]$, for every vertex $v$ and color $c$, but rather we will focus on their extreme values.
In particular, we will define appropriate $L_i, T_{i,r}$ such that for each $i$ the following property holds in the beginning of iteration $i$:

\smallskip
\noindent
\emph{Property P(i):} For each  vertex $v \in V_i$,  color $c \in L_v$ and $ r \in [k-1]$:
\begin{eqnarray*}
\ell_i(v) & \ge & L_i;\\
t_{i,r} (v,c) & \le & T_{i,r}.
\end{eqnarray*}

As a matter of fact, it would be helpful for our analysis (though not necessary) if the inequalities defined in $P(i)$ were actually tight. Given that $P(i)$ holds, we can always enforce this stronger property in a straightforward way as follows. First,  for each
vertex $v$ such that $\ell_i(v) > L_i$ we choose arbitrarily $\ell_i(v) - L_i$ colors from its list and remove them.  Then, for each vertex $v$ and color $c \in L_i$  such that $t_{i,r}(v,c)  < T_{i,r}$ we add  to the hypergraph $T_{i,r} - t_{i,r}(v,c)$   new hyperedges of size $r+1$ that contain $v$ and $r$ new ``dummy" vertices.  (As it will be evident from the proof, we can always assume that $L_i, T_{i,r}$  are integers, since our analysis is robust to replacing $L_i, T_{i,r} $ with $\lfloor L_i \rfloor $ and $T_{i,r} $ with $\lceil T_{i,r} \rceil$.) We assign each dummy vertex a list of $L_i$  colors: $L_i-1$ of them are  new and  do not appear in the list of any other vertex, and the last one is $c$. 

\begin{remark}
Dummy vertices are only useful for the purposes of our analysis and can be removed at the end of the iteration. Indeed, one could use the technique of ``equalizing coin flips" instead. For more details  see e.g.,~\cite{mike_book}.
\end{remark}
Overall, without loss of generality, at each iteration $i$ our goal will be to guarantee that Property $P(i+1)$ holds assuming  Property $Q(i)$.

\smallskip
\noindent
\emph{Property Q(i):} For each vertex $v \in V_i$,  color $c \in L_v$ and $ r \in [k-1]$:
\begin{eqnarray*}
\ell_i(v) & =  & L_i;\\
t_{i,r} (v,c) & =  & T_{i,r}.
\end{eqnarray*}

\paragraph{An iteration.} 
For the $i$-th iteration, consider the probability distribution induced by assigning each vertex $v \in V_i$ a color chosen uniformly at random from $L_v$ and activating $v$ with probability $\alpha = \frac{K}{\ln \Delta}$, where $K = (100 k^{3k})^{-1}$. From an existential perspective,  we will show that the Lovász Local Lemma may be applied to demonstrate that a sample from this distribution yields a state $\sigma_i$ satisfying the required conditions — specifically Property $P(i+1)$ — with strictly positive probability. Computationally, the Moser–Tardos algorithmic framework ensures that such a state $\sigma_i$ can be found in polynomial time.

When a desired state $\sigma_i$ is obtained according to the above process,  the criteria for color retention, and the transition rules for updating the candidate lists $L_v$ are as follows. The list of each vertex $v $, $L_v(i+1)$, is induced from $L_v(i)$ by removing every non-available color $c  \in L_v(i)$ for $v$ in $\sigma_i$. A vertex $v$ retains its assigned color if it is activated and its assigned color is available under $\sigma_i$.

Overall, the $i$-th iteration of our coloring algorithm can be described at a high-level as follows.

\begin{enumerate}

\item Apply the Moser–Tardos algorithm in the probability space induced by a uniform random coloring of each vertex $v \in V_i$ from $L_v(i)$ and the independent activation  of each vertex with probability $\alpha$, targeting the avoidance of a set of bad events whose non-occurrence ensures that Property $P(i+1)$ holds. Let $\sigma_i$ be the output state of the Moser-Tardos algorithm.

\item For each vertex $v \in V_i$, remove any non-available color $c\in L_v(i)$  in $\sigma_i$ to get a list $L_v(i+1)$.

\item Uncolor every vertex $v \in V_i$ that has either received a non-available color or is deactivated  under $\sigma_i$, to get a new partial list-coloring $\phi_i$.

\end{enumerate}

\paragraph{Controlling the parameters of interest.} 
Next we describe the recursive definitions for $L_i$ and $T_{i,r}$ which, as we already explained,  will determine the behavior of the parameters $\ell_i(v)$ and $t_{i,r}(v,c)$, respectively.

Initially, $L_1 = (1+\delta ) \left( \frac{ \Delta }{ \ln \Delta }  \right)^{ \frac{1}{k-1} }  $, $T_{1,k-1} = \Delta$ and $T_{1,r}  = 0 $ for every $ r \in [k-2]$. Letting
\begin{align}\label{Keep_def}
\mathrm{Keep}_i = \prod_{ r =1}^{k-1 } \left(1 -  \left(  \frac{\alpha}{L_i } \right)^{ r}   \right)^{T_{i,r} },
\end{align}
we define
\begin{eqnarray}
L_{i+1} & =  & L_i \cdot \mathrm{Keep}_i - L_i^{2/3},\label{L_def} \\
T_{i+1, r} & = &  \sum_{j =  r }^{ k-1 } \left(  T_{i,j} \cdot { j \choose  r } \left(   \mathrm{Keep}_i \left(  1- \alpha \mathrm{Keep}_i \right)   \right)^{ r }   \left(\frac{ \alpha \mathrm{Keep}_i}{L_i  }\right)^{j-r}   \right)  \nonumber \\
		&  &+  4 k^{ 2(k-r) } \alpha  (\alpha^{-1}  L_i)^{r}\ln \Delta        \sum_{\ell=1}^{k-1}\frac{T_{i,\ell} }{ L_i^{2\ell} (\ln \Delta)^{2\ell} }    +  \left( \sum_{j=r}^{k-1}  { j \choose r } \alpha^{j-r} \frac{T_{i,j}}{L_i^{j-r} } \right)^{2/3}			 \label{T_def}.
\end{eqnarray}

To get some intuition for the recursive definitions~\eqref{L_def},~\eqref{T_def}, observe that $\mathrm{Keep}_i$ is the probability that a color $c \in L_v(i)$ is present in $L_v(i+1)$ as well.  Note  further that this implies that  the expected value of $\ell_{i+1}(v)$  is   $L_i \cdot \mathrm{Keep}_i$, a fact which motivates~\eqref{L_def}. Calculations of similar flavor for $\ex[ t_{i+1,r}(v,c) ]$ motivate~\eqref{T_def}.

\paragraph{The key lemmas.} 

We are almost ready to state the main lemmas which  guarantee that our procedure eventually reaches a partial list-coloring of $H$ with favorable properties that  allow us to extend it to  a full list-coloring. Before doing so, we need to settle a subtle issue that has to do with the fact that  $t_{i+1, r}(v,c)$ is not sufficiently concentrated around its expectation, as it has high sensitivity to the color assignment of $v$. To see this, notice for example that $t_{i+1,1}(v,c)$ drops to zero if $v$ is  activated and is assigned $c$ in state $\sigma_i$. This is because then $c$ becomes unavailable for all vertices in $D_{i,1}(v,c)$, and therefore it has to be removed from their list of candidate colors for the next iteration. More generally, for $r \in \{2, \ldots, k-1\}$, if $v$ is assigned $c$ then $t_{i+1, r}(v,c)$ can be affected by a large amount. As an example, consider a hyperedge $h \in D_{i,r}(v,c)$. Suppose that among the $r$ uncolored vertices in $h \setminus \{v\}$, a specific subset of $r-1$ vertices are assigned color $c$  and are activated in $\sigma_i$. In this scenario, the status of $h$ for the next iteration depends entirely on $v$: if $v$ is also  activated and assigned color $c$, the remaining uncolored vertex in $h$ loses $c$ as a candidate, and $h$ is removed from $D_{i+1,r}(v,c)$. Because $v$ can be the ``deciding factor" for many such hyperedges simultaneously, $t_{i+1,r}(v,c)$ can fluctuate wildly based on $v$'s color assignment.

To deal with this problem, we will focus instead on variable $t_{i+1,r}'(v,c)$, i.e., the number of hyperedges $h$ that contain $v$ and (i) exactly $k -r -1$ vertices of $h \setminus \{v \}$ are colored $c$ in the end of iteration $i$ (either because they were already  colored from previous iterations, or because they were assigned color $c$ during iteration $i$ and retained it); (ii) the remaining $r$ vertices of $h \setminus \{v \} $ did not retain a color and, crucially, $c$ would be available for  them if we ignored the color of $v$ in $\sigma_i$. Observe that if $c$  is not assigned to $v$ then $t_{i+1, r}(v,c) = t_{i+1,r}^{'  } (v,c)$ and $t_{i+1,r}'(v,c) \ge  t_{i+1,r} (v,c)$ otherwise.

The first lemma that we prove  estimates the expected value of the parameters at the end of the $i$-th iteration. Its proof can be found in Section~\ref{omitted_hyper}.

\begin{lemma}\label{expectations_lemma}
Let $S_i = \sum_{\ell =1}^{k-1 } \frac{T_{i,\ell} }{ L_i^{2\ell} (\ln \Delta)^{2\ell} }  $ and $Y_{i,r} = \sum_{j = r}^{k-1} \frac{T_{i,j} }{ L_{i}^{j}} $. If $Q(i)$ holds and for all $1 < j < i, r \in [k-1],  L_j \ge (\ln \Delta)^{20(k-1)}, T_{i,r} \ge (\ln \Delta)^{20(k-1)} $, then, for every vertex $v \in V_{i+1}$ and color $ c \in L_v$:
\begin{enumerate}[(a)]
\item  $\ex[ \ell_{i+1}(v)  ]  = \ell_i(v) \cdot \mathrm{Keep}_i $;\label{ex_part_a}
\item  \begin{align*}\ex[ t_{i+1,r}' (v, c) ] \le &\sum_{j =r }^{k-1 }  \left(  T_{i,j}\cdot  { j \choose r} \left( \mathrm{Keep}_i  \left(1 - \alpha \mathrm{Keep}_i  \right)\right)^{r}   \left( \frac{ \alpha \mathrm{Keep}_i }{L_i } \right)^{j-r} \right)   \\
   &+ 4 k^{ 2(k-r) } \alpha  (\alpha^{-1}  L_i)^{r}      S_i   \ln \Delta + O(Y_{i,r} ).
\end{align*}
\label{ex_part_b} 
\end{enumerate}
\end{lemma}

The next step is to prove strong concentration around the mean for our random variables per the following lemma. Its proof can be found in Section~\ref{omitted_hyper}.

\begin{lemma}\label{concentration_lemma}
If $Q(i)$ holds, $L_i, T_{i,j} \ge (\ln \Delta)^{20(k-1)}$ for all  $j \in [k-1]$, and $T_{i,k-1} \ge \frac{1}{10k^2 }  L_j^{k-1}$,  then for every vertex $v \in V_{i+1}$, color $ c \in L_v$ and $r \in [k-1]$:
\begin{enumerate}[(a)]
\item $\Pr\left[ | \ell_{i+1}(v) - \ex[ \ell_{i+1}(v) ]      | < L_i^{2/3}    \right] <  \Delta^{ - \ln \Delta } $; \label{con_part_a}

\item $\Pr\left[  t_{i+1,r}'(v,c) - \ex[ t_{i+1,r}' (v,c) ]     > \frac{1}{2} \left( \sum_{j=r}^{k-1}  { j \choose r } \alpha^{j-r} \frac{T_{i,j}}{L_i^{j-r} } \right)^{2/3}	   \right] <  \Delta^{ - \ln \Delta } $. \label{con_part_b}
\end{enumerate}
\end{lemma}

Armed with Lemmas~\ref{expectations_lemma},~\ref{concentration_lemma},  a straightforward application of  the symmetric  Local Lemma, i.e.,  Corollary~\ref{symmetric_LLL}, reveals the following. 

\begin{lemma}\label{final_corollary}
With positive probability, $P(i)$ holds for every $i$ such that for all $1 < j < i: L_j, T_{j,r} \ge (\ln \Delta)^{20(k-1)}$ for all $r \in [k-1]$ and $T_{j,k-1} \ge \frac{1}{10k^2} L_j^{k-1} $.
\end{lemma}
The proof of Lemma~\ref{final_corollary} can be found in Section~\ref{omitted_hyper}.

In analyzing the recursive equations~\eqref{L_def},~\eqref{T_def}, it would be helpful if we could ignore the ``error terms". The next lemma shows that this is indeed possible. Its proof can be found in Section~\ref{omitted_hyper}.
\begin{lemma}\label{no_errors}
Define $L_1' = (1+\delta) \left( \frac{ \Delta}{  \ln \Delta} \right)^{ \frac{1}{k-1}}, T_{1,k-1}' = \Delta $, $T_{1,r}' = 0$ for $r \in [k-2]$,  and recursively define
\begin{eqnarray}
L_{i+1}' & =& L_i' \cdot \mathrm{Keep}_i, \label{Lprime_def} \\
T_{i+1,r}' & =& \sum_{j =  r }^{ k-1 } \left(  T_{i,j}' \cdot { j \choose r}  \left(  \mathrm{Keep}_i  \cdot \left( 1-  \alpha \mathrm{Keep}_i  \right)  \right)^{ r }    \left( \frac{ \alpha \mathrm{Keep}_i }{L_i'  } \right)^{j-r}   \right)  \nonumber \\
	       &  &+   4 k^{ 2(k-r) } \alpha  (\alpha^{-1}  L_i)^{r}\ln \Delta        \sum_{\ell=1}^{k-1}\frac{T_{i,\ell} }{ L_i^{2\ell} (\ln \Delta)^{2\ell} }  . \label{Tprime_def}
\end{eqnarray}
If for all $ 1 < j < i  $, $L_j \ge (\ln \Delta)^{20(k-1)}$, $T_{j,r} \ge (\ln \Delta)^{20(k-1)}$ for every $r \in [k-1]$,  and $T_{j,k-1} \ge \frac{  L_j^{k-1}}{ 10k^2 }$, then 
\begin{enumerate}[(a)]
\item $| L_i - L_i' | \le (L_i')^{\frac{5}{6} }$; 
\item $|T_{i,r} - T_{i,r}' | \le (T_{i,r}')^{ \frac{ 100r}{ 100r +1}   } $.
\end{enumerate}
\end{lemma}

\begin{remark}
Note that $\mathrm{Keep}_i$ in Lemma~\ref{no_errors}  is still defined in terms of $L_i,T_{i,r}$  and not $L_i',T_{i,r}'$. Note also that in the definition of $T_{i+1,r}'$, the second summand is a function of $T_{i,\ell}, L_{i}$, $\ell \in [r-1]$, and not $T_{i,\ell}', L_i'$.
\end{remark}

Using Lemma~\ref{no_errors} we are able to prove the following in Section~\ref{omitted_hyper}.
\begin{lemma}\label{target_lemma}
There exists $i^* = O (\ln \Delta   \ln \ln \Delta)$ such that
\begin{enumerate}[(a)]
\item For all $1 < i \le i^*, T_{i,r} > (\ln \Delta)^{20(k-1)}, L_i \ge \Delta^{  \frac{\epsilon/3 }{ (k-1)  (1 + \epsilon/2 ) }    } $, and $T_{i,k-1} \ge \frac{1}{10 k^2} L_i^{k-1}$; 
\item $T_{i^*+1,r} \le  \frac{1}{10k^2} L_{i^* +1 }^{r} $, for every $r \in [k-1]$ and $L_{i^*+1}  \ge \Delta^{ \frac{\epsilon/3 }{ (k-1)  (1 + \epsilon/2 ) }   }$.
\end{enumerate}
\end{lemma}

Lemmas~\ref{final_corollary},~\ref{target_lemma} and~\ref{second_phase} imply Theorem~\ref{main_hypergraph}. 
\begin{lemma}\label{second_phase}
Let $i^*$ be the integer promised by Lemma~\ref{target_lemma}, and assume property $P(i^*+1)$ holds for the partial list-coloring $\phi_{i^*}$, i.e., the output of the $i^*$-iteration which is the partial list-coloring in the beginning of the $(i^*+1)$-th iteration. Given $\phi_{i^*}$,  we can find a full list-coloring of $H$ in expected polynomial time in  the number of vertices of $H$.  Also, if $\Delta $ is assumed to be constant, then such a coloring can be constructed deterministically in polynomial time.
\end{lemma}

\begin{proof}[Proof of Theorem~\ref{main_hypergraph}]
We carry out  $i^*$ iterations of our procedure. If $P(i)$ fails to hold for any iteration $i$, then we halt. By Lemmas~\ref{final_corollary} and~\ref{target_lemma}, $P(i)$ (and, therefore, $Q(i)$) holds with positive probability for each iteration and so it is possible to perform $i^*$ iterations. Further,  since the application of the Lov\'{a}sz Local Lemma in the proof of Lemma~\ref{final_corollary}  is within the scope of the variable setting, Theorem~\ref{thm:randomized_mt} applies and the Moser-Tardos algorithm terminates in expected polynomial time.  In particular, recall that the proof of Lemma~\ref{final_corollary} employs the symmetric version of the Lovász Local Lemma (Corollary~\ref{symmetric_LLL}). In this application, the probability of each bad event is bounded by $p := \Delta^{-\ln \Delta}$, while the size of each dependency neighborhood is bounded by $d := \Delta^5$. Consequently, Theorem~\ref{thm:randomized_mt} and Remark~\ref{llls} provide a polynomial upper bound on the expected running time of the Moser-Tardos algorithm. Furthermore, for sufficiently large $\Delta$, Theorem~\ref{harris_theorem_derandomize} ensures that the process can be derandomized to yield a deterministic polynomial-time algorithm.

Thus, we can execute $i^*$ successful iterations in polynomial time. Following these iterations, we apply the algorithm from Lemma~\ref{second_phase} to complete the list-coloring. For constant $\Delta$, this second phase can also be efficiently derandomized.

\end{proof}

\subsection{Proof of Lemma~\ref{second_phase}}\label{second_phase_proof}

To lighten the notation, let $\phi = \phi_{i^*}$.  Let  also  $\mathcal{U}_{\phi}$ denote the set of uncolored vertices in $\phi$, and   $\mathcal{U}_{\phi}(h) $ the subset of $\mathcal{U}_{\phi}$ that belongs to a hyperedge $h$.  Our goal is to color the vertices in $\mathcal{U}_{\phi}$ to get a proper list-coloring. 

Towards that end, let $L_v = L_v(\phi)$ denote the list of colors for $v$ in $\phi$, and $D_r(v,c) := D_{i^*+1,r}(v,c)$ the set of hyperedges (of size $t_{i^*+1,r}(v,c)$) with $r$ uncolored vertices in $\phi$ whose vertices ``compete" for $c$ with $v$,  and recall the conclusion of Lemma~\ref{target_lemma}. Let $\mu$ be the probability distribution induced by giving each vertex $v \in  \mathcal{U}_{\phi}$ a color from $L_v$ uniformly at random.  For every hyperedge $h$ and color $c$ such that  (i) $c \in \bigcap_{ v \in \mathcal{U}_{\phi}(h) } L_v $; and (ii) $\phi(v) = c$ for every vertex in $ h \setminus \mathcal{U}_{\phi}(h) $,  we define $A_{h,c}$ to be the event that all vertices of $h$ are colored $c$. Let $\mathcal{A}$ be the family of these (bad) events, and observe that any elementary event (list-coloring) that does not belong in their union is a proper. In other words, if we avoid these bad events we have found a proper list-coloring of the hypergraph. Moreover, for every $A_{h,c} \in \mathcal{A}$:
\begin{align*}
\mu\left( A_{h,c} \right)  \le \frac{1}{ \prod_{ v \in  \mathcal{U}_{\phi}(h) } |L_v(\phi)|   }  < \frac{1}{4},
\end{align*}
for large enough $\Delta$, since $ L_{i^*+1}= L_{i^*+1}(\Delta) \xrightarrow{  \Delta \to +\infty}  +\infty$.

Define 
\begin{align*}
D(A_{h,c} ) := \bigcup_{v \in \mathcal{U}_{\phi}(h) } \bigcup_{ c' \in L_v}   \bigcup_{r =1}^{k-1}  \left\{A_{h',c'} : h' \in D_r(v,c')\right\}
\end{align*}
and observe that $A_{h,c}$ is mutually independent of the events in $\mathcal{A} \setminus D(A_{h,c} )$.  The existential claim of Lemma~\ref{second_phase}  follows from Corollary~\ref{MikeLLL}  as, for every $A_{h,c} \in \mathcal{A}$:
\begin{eqnarray}
\sum_{ A \in D(A_{h,c} ) }  \mu(A) 	& \le &	 \sum_{ v \in \mathcal{U}_{\phi}(h) } \sum_{ c' \in L_v } \sum_{r = 1}^{k-1} \sum_{ h'\in D_{r}(v,c')}  \mu \left( A_{ h',c'} \right) \nonumber \\
 	& =  &   \sum_{ v \in \mathcal{U}_{\phi}(h) } \sum_{ c' \in L_v } \sum_{ i=1}^{k-1} \sum_{ h' \in D_{r}(v,c')}  \frac{1}{ \prod_{u \in \mathcal{U}_{\phi} (h') } |L_u| }   \nonumber \\
 	 & \le  &   \max_{ v\in \mathcal{U}_{\phi}(h)  } \frac{k}{   | L_v|   }    \sum_{ c' \in L_v  } \sum_{ r=1 }^{k-1}  \frac{ | D_{r}(v,c') | }{  L_{i^*+1}^{r} }   \label{one_second}\\
	 & \le&   \frac{k}{10k^2} \max_{ v \in \mathcal{U}_{ \phi }(h)  }   \frac{L_{i^* +1 }^{r}  \cdot  |L_v |}{|L_v| \cdot   L_{i^*+1}^{r} }   \label{two_second} \\
	 &   \le & \frac{1}{10}   < \frac{1}{4} \label{three_second},
 \end{eqnarray}
concluding the proof. Note that in~\eqref{one_second} we used the facts that every hyperedge has at most $k$ vertices and $L_{i^*+1}  \ge \Delta^{ \frac{\epsilon/3 }{ (k-1)  (1 + \epsilon/2 ) }   }$, and in~\eqref{two_second} we used the fact that 
 $|D_{r}(v,c') | \le T_{i^*+1}^{r} \le  \frac{1}{ 10k^2} L_{i^*+1 }^{r}  $.

Regarding the algorithmic claim, since our application of the Lov\'{a}sz Local Lemma fits the variable setting, the Moser-Tardos algorithm is applicable (specifically via Theorem~\ref{thm:randomized_mt} and Remark~\ref{llls}) and terminates in expected polynomial time. Furthermore, if $\Delta$ is assumed to be constant, Theorem~\ref{thm:deterministic_mt} ensures that the process can be derandomized to yield a deterministic polynomial-time algorithm. This is applicable because the specific structure of our bad events allows for the efficient computation of conditional probabilities. Additionally,~\eqref{three_second} provides the constant slack necessary to satisfy the requirements of Theorem~\ref{thm:deterministic_mt}.

\section{Hypergraph list-coloring proofs}\label{omitted_hyper}

In this section we prove  Lemmas~\ref{expectations_lemma},~\ref{concentration_lemma},~\ref{final_corollary},~\ref{no_errors},~\ref{target_lemma}.

We start by stating a couple of  important technical lemmas that will be helpful for these proofs. To streamline the presentation, we reserve their detailed proofs for the end of the section. It will be convenient to define $R_{i,r} = \frac{T_{i,r} }{ L_i^r } $, $R_{i,r}' = \frac{T_{i,r}' }{ (L_i')^r}  $ for every $r \in [k-1]$.

\begin{lemma}\label{bounding_keep_lemma}
If  for all $1 < j < i, r \in [k-1],  L_j, T_{j,r}  \ge (\ln \Delta)^{20(k-1)}$, then  
\begin{align*}
R_{i,r}      \le    k^{2( k-1-r) }   \ln \Delta.
\end{align*}
\end{lemma}
The proof of Lemma~\ref{bounding_keep_lemma} can be found in Subsection~\ref{bounding_keep_lemma_proof}. A straightforward corollary of Lemma~\ref{bounding_keep_lemma} is the following.
\begin{corollary}\label{bounding_keep}
If $L_i, T_{i,r} \ge (\ln \Delta)^{20(k-1)}$ and $R_{i,k-1} \ge \frac{ 1}{ 10k^2 }$, then 
\begin{align*}
C:= \mathrm{exp}\left( - \frac{ K k^{ 2(k-2)  } }{1 -  \frac{\delta}{100k} }      \right) \le \mathrm{Keep}_i \le  1 -  \frac{K^{k-1}}{12 k^2  ( \ln \Delta)^{k-1} }.
\end{align*}
\end{corollary}
\begin{proof}
The lower bound follows directly from~\eqref{constant_bound} (which appears in the proof of Lemma~\ref{bounding_keep_lemma}). The upper bound follows from our assumption that  $R_{i,k-1} \ge \frac{ 1}{ 10k^2 }$ which implies  that
\begin{align*}
\mathrm{Keep}_i     \le \mathrm{e}^{ - \sum_{r=1}^{k-1}  \alpha^{r}  R_{i,r} }  \le \mathrm{e}^{-  \alpha^{k-1}R_{i,k-1}  }  \le \mathrm{e}^{ - \frac{ K^{k-1}   }{ 10 k^2   (\ln \Delta)^{k-1} } }  <  1 -  \frac{K^{k-1}}{12 k^2  (\ln \Delta )^{k-1 } },
\end{align*}
for sufficiently large $\Delta$. 
\end{proof}

The proof of the following lemma can be found in Subsection~\ref{pain_lemma_proof}.
\begin{lemma}\label{pain}
If  $L_j, T_{j,r} \ge (\ln \Delta)^{20(k-1)}$ for all $1 < j < i$, then for every $r \in [k-1]$:
\begin{align*}
R_{i,r }' \le  (1 - \alpha C )^{r (i-1)} \ln \Delta   \cdot     \frac{   (1+  \frac{\delta}{k^{100}} )^{k-1-r}  }{ (1+\delta - \frac{\delta}{k^{99} } )^{k-1} C^{k-1-r }   } \prod_{p=r}^{ k-2}(p+1).
\end{align*}
\end{lemma}

We are now ready to prove Lemmas~\ref{expectations_lemma},~\ref{concentration_lemma},~\ref{final_corollary},~\ref{no_errors} and~\ref{target_lemma}.

\subsection{Proof of Lemma~\ref{expectations_lemma}} 

\begin{proof}[Proof of part~\eqref{ex_part_a}] 
For every color $c \in L_v(i)$,
\begin{align}\label{quick_calc}
\Pr[ c \in L_v(i+1)  ]  =  \prod_{r=1}^{k-1}  \prod_{ h \in D_{i,r}(v,c)  }  \left( 1-  \prod_{u \in (h  \setminus \{v \} )\cap V_i  } \frac{  \alpha}{ \ell_i(u)  }  	\right) =  \prod_{r=1}^{k-1}  \left(  1-  \left(\frac{   \alpha}{  L_i}\right)^r 	\right)^{T_{i,r}} = \mathrm{Keep}_i,
\end{align}
where for the second equality we used our assumption that $Q(i)$ holds. Therefore, the proof of the first part of the lemma follows from the  linearity of expectation.

\end{proof}

\begin{proof}[Proof of part~\eqref{ex_part_b} ]
Recall the definition  of $t_{i+1,r}'(v,c)$ and  note that   only hyperedges  in $\bigcup_{j = r}^{k-1}  D_{i,j}(v,c) $  can be potentially counted by $t_{i+1,r}'(v,c)$. In particular, unless $v$ and $j-1$ other uncolored vertices of a hyperedge  $h \in D_{i,j}(v,c)$, $j \ge r$, are assigned $c$ during iteration $i$, then if $h$ is counted by $t_{i+1,r}'(v,c)$, it is also counted by $t_{i+1,r}(v,c)$. Therefore,
\begin{align}\label{three_one_first}
\ex[ t_{i+1,r}'(v,c)  ] \le \ex[ t_{i+1,r}(v,c)  ]  + O\left( \sum_{j=r}^{k-1} \frac{T_{i,j}}{L_i^j}\right), 
\end{align}  
and so we  focus on bounding $\ex[ t_{i+1,r}(v,c)  ]$.

Fix $h \in D_{i,j}(v,c)$, where $j \ge r$. Our goal will be to show that 
\begin{align}\label{tornado_warning} 
\Pr[ h  \in D_{i+1,r}(v,c) ] \le&   { j \choose r }  \left( \mathrm{Keep}_i (1-\alpha \mathrm{Keep}_i) \right)^r    \left(  \frac{ \alpha \mathrm{Keep}_i }{ L_i }  \right)^{j-r}    \nonumber \\	
				 &+4 r {j \choose r} \frac{ \mathrm{Keep}_i^{j-1} \alpha^{ j-r+1} S_i }{ L_i^{j-r} } 			 +O\left(  \frac{1}{ L_i^j } \right),
\end{align} 
since combining~\eqref{tornado_warning} with~\eqref{three_one_first} implies the lemma.   To see this, observe that
\begin{align}
T_{i,j} \cdot 4 r {j \choose r} \frac{ \mathrm{Keep}_i^{j-1} \alpha^{ j-r+1} S_i }{ L_i^{j-r} } 	&=   4 \alpha r {j \choose r }  \cdot \alpha^j \frac{ T_{i,j}  }{ L_{i}^{ j}  }  \mathrm{Keep}_i^{j-1} \cdot  (\alpha^{-1}  L_i)^{r} S_i \nonumber   \\
&\le 
 \begin{cases} \label{trick_trick}
         4  T_{i,1}  S_i ,
 &\text{if }  j=r =1 , \\
       \frac{4  \alpha r {j \choose r }}{\mathrm{e} (j-1) }   (\alpha^{-1}  L_i)^{r} S_i  & \text{otherwise.}
   \end{cases}
\end{align}
Note that in deriving the second part of the inequality in~\eqref{trick_trick} we first used  that $1 - x \le \mathrm{e}^{-x} $ for every $ x \ge 0$  in order to bound  $\mathrm{Keep}_i$ by $\mathrm{exp}(- \alpha^j T_{i,j} /L_i^j )$ ,  and then  that $\max_{x} x \mathrm{e}^{ - \ell x} \le \frac{1 }{ \ell  \mathrm{e} }  $ for every $\ell $. Therefore, 
\begin{eqnarray}
\ex[ t_{i+1,r}(v,c)  ]  &\le & \sum_{j=r}^{k-1 } T_{i,j} \max_{ h \in D_{i,j}(v,c) }   \Pr[ h  \in D_{i+1,r}(v,c) ]  \\
			     & < &    \sum_{j =r }^{k-1 }  \left(  T_{i,j}\cdot  { j \choose r} \left( \mathrm{Keep}_i  \left(1 - \alpha \mathrm{Keep}_i  \right)\right)^{r}   \left( \frac{ \alpha \mathrm{Keep}_i }{L_i } \right)^{j-r} \right)    \nonumber      \\ 
			     &  &            +  4 k^{ 2(k-r) } \alpha  (\alpha^{-1}  L_i)^{r} S_i \ln \Delta  + O\left( \sum_{j=r}^{k-1} \frac{T_{i,j} }{ L_i^{j} }   \right), \label{koura}
\end{eqnarray}
for sufficiently large $\Delta$. In deriving~\eqref{koura} we used~\eqref{trick_trick} and the facts that:
\begin{align*}
 T_{i,1}  & =  \frac{ T_{i,1}}{L_i}   \cdot L_i \le L_i  \cdot k^{2(k-1-r) }\ln \Delta,  \enspace \mbox{           according to Lemma~\ref{bounding_keep_lemma};}  \\
\sum_{j = r}^{k-1} \frac{ r {j \choose r }} { \mathrm{e} (j-1)  }     & <  \ln \Delta \cdot  k^{2(k-1-r)  }  \mbox{ for sufficiently large $\Delta$ and $r > 1$.} 
\end{align*}

Towards proving~\eqref{tornado_warning},  for any vertex $u \in h \setminus \{v \}$, consider the events
\begin{eqnarray*}
E_{u,1} & =&  \text{ ``$ u $ does not retain its color  and $ c \in L_u(i+1) $"}, \\
E_{u,2} & =&  \text{ ``$u$ is assigned $c$ and retains its color"}. 
\end{eqnarray*}
Let also $B_c$ be the event that $v$ and $j-1$ other uncolored vertices of $h$ receive color $c$ in $\sigma_i$. Since we have assumed that our hypergraph is of girth at least $5$ (and thus at least 4, which  suffices here), for any neighbor $u$ of $v$ and $f \in \{1,2\}$ the event  $E_{u,f}$ is mutually independent of all events $E_{u', \ell}, \ell \in \{1,2\}, u \ne u'$, conditional on $B_c$ not occurring.  Thus, if $\Pr[ E_{u,\ell} \mid  \overline{B_c}  ] \le p_{\ell}$, $\ell \in \{1,2\}$,   for every vertex $u \in h \setminus \{v \}$, we obtain
\begin{align}\label{three_two_second}
\Pr[ h  \in D_{i+1,j}(v,c) ]\le      { j \choose  r } p_{1}^{r}   p_2^{j-r}  +  \Pr[B_c] \le { j \choose  r } p_{1}^{r}   p_2^{j-r} +  \frac{2^k}{ L_i^j}, 
\end{align}
since $\Pr[ B_c] \le   2^kL_i^{-j}$.

Now we claim that for any  $u \in h \setminus \{v \}$, and sufficiently large $\Delta$,
\begin{align}\label{E2_bound}
\Pr[ E_{u,2} \mid \overline{B_c}] \le   \frac{ \alpha \mathrm{Keep}_i }{L_i } + \frac{2 }{ (L_i  \ln \Delta) ^{j+1} } =:    q_2  + \delta_2.
\end{align}
To see this, notice that  conditional on $\overline{B_c}$ the probability that   $u$   is activated is $\alpha$, it is assigned $c$ with probability  at most  $ 1/L_i$, and  it retains $c$  with probability that is at most
\begin{align}\label{same_same_same}
 \prod_{r  \in [k-1] \setminus \{j \} } \left(1 - \frac{\alpha^r}{L_i^r } \right)^{ T_{i,r} } \cdot \left( 1- \frac{ \alpha^j  }{ L_i^j }  \right)^{T_{i,j} - 1 }   = \frac{ \mathrm{Keep}_i }{1 - \frac{\alpha^j}{L_i^j }  }.
\end{align}
Thus,
\begin{align*}
\Pr[ E_{u,2} \mid \overline{B_c}]  \le  \frac{ \alpha \mathrm{Keep}_i }{L_i( 1- \frac{ \alpha^j }{L_i^j } )  }  \le \frac{\alpha \cdot \mathrm{Keep}_i }{L_i } \cdot \left(1 + \frac{2 \alpha^j}{ L_i^j}  \right) \le  \frac{ \alpha \mathrm{Keep}_i }{L_i } + \frac{2 }{ (L_i  \ln \Delta) ^{j+1} }.
\end{align*}
for sufficiently large $\Delta$, concluding the proof of~\eqref{E2_bound}.

Further, we  claim that
\begin{align}\label{E1_bound}
\Pr[ E_{u,1} \mid \overline{B_c}] \le   \mathrm{Keep}_i( 1- \alpha \mathrm{Keep}_i)  + \left(  2\alpha  S_i +  (L_i \ln \Delta)^{-j} ( 3 + 4 \alpha S_i ) \right) =: q_1 + \delta_1.
\end{align}

To show~\eqref{E1_bound} we consider three cases. The first case is  that $u$ is not activated and  $c \in L_u(i+1)$ (notice that these are two independent events). In this case $u$ will not retain its  color, and observe that
\begin{align}\label{alex_h}
\Pr[ c \in L_u(i+1)  \mid \overline{B_c}] \le \frac{  \mathrm{Keep}_i }{ 1- \frac{ \alpha^j }{L_i^j }   }  \le  \mathrm{Keep}_i  \left( 1+ \frac{2 \alpha^j }{L_i^j }\right)    \le \mathrm{Keep}_i + 2 ( L_i  \ln \Delta ) ^{-j}.
\end{align}
Thus,
\begin{align}\label{first_casara}
\Pr[ \text{$u$ is not activated and $c \in L_u(i+1) $}  \mid \overline{B_c} ] \le (1-\alpha) \left( \mathrm{Keep}_i + 2 ( L_i  \ln \Delta )^{-j} \right).
\end{align}
In the second case we consider the scenario where  $u$ is activated and is assigned $c$ in $\sigma_i$. Clearly then, the probability that $c \in L_u(i+1)$ and $u$ does not retain $c$ is zero.   Finally, suppose that $u$ is activated and is assigned  a color $\gamma \ne c$ in $\sigma_i$. Our goal is to compute
$\Pr[ \text{($u$ is activated and assigned $\gamma$)} \wedge E_{u,1} \mid \overline{B_c} ]$ for each $\gamma$ so that we can sum up these probabilities over all possible $\gamma \ne c$ along with~\eqref{first_casara}.

For a vertex $w$ let $F_w^{\gamma}$ denote the event that $w$ is activated and assigned $\gamma$ in $\sigma_i$. Using this notation we have:
\begin{align}
 \Pr[ F_u^{\gamma}  \wedge E_{u,1} \mid \overline{B_c} ] & =\Pr[ F_u^{\gamma} \mid    \overline{B_c} ]  \cdot  \Pr[(\gamma \notin L_u(i+1) ) \wedge (c \in L_u(i+1))  \mid F_u^{\gamma}, \overline{B_c}      ]  \nonumber \\
																	&= \frac{\alpha}{L_i}  \cdot \Pr[(\gamma \notin L_u(i+1) ) \wedge (c \in L_u(i+1))  \mid  F_u^{\gamma} , \overline{B_c} ] \nonumber \\
																	& =  \frac{\alpha}{L_i} \Pr[  c \in L_u(i+1) \mid F_u^{\gamma}, \overline{B_c}   ] \cdot \Pr[ \gamma \not \in L_u(i+1) \mid c \in L_u(i+1), F_u^{\gamma}, \overline{B_c}  ] \nonumber \\
																	& =  \frac{\alpha}{L_i} \Pr[  c \in L_u(i+1) \mid \overline{B_c}   ] \cdot \Pr[ \gamma \not \in L_u(i+1) \mid c \in L_u(i+1), F_u^{\gamma}, \overline{B_c}  ] \nonumber \\
																	& \le   \frac{\alpha}{L_i} \cdot (\mathrm{Keep}_i + 2 ( L_i  \ln \Delta ) ^{-j} )\cdot  \Pr[ \gamma \not \in L_u(i+1) \mid c \in L_u(i+1), F_u^{\gamma}, \overline{B_c}  ] \label{just_you_wait}
\end{align}
and so below we focus on  bounding  for each $ \gamma \in L_u(i) \setminus \{ c\}$  the probability that  $ \gamma \notin L_u(i+1)$     conditional on  that $c \in L_u(i+1)$,  $u$ is activated and assigned $\gamma$ in $\sigma_i$, and $B_c$ did not occur. Note that in deriving~\eqref{just_you_wait} we used~\eqref{alex_h}.

We have:
\begin{align}
&\Pr[ \gamma \not \in L_u(i+1) \mid c \in L_u(i+1), F_u^{\gamma}, \overline{B_c}  ]   = 1 - \Pr[ \gamma \in L_u(i+1) \mid c \in L_u(i+1), F_u^{\gamma}, \overline{B_c} ] 			\nonumber \\
 & = 1 -  \prod_{\ell =1}^{k-1} \prod_{ g \in D_{i,\ell}(u,\gamma)}   \left( 1 - \Pr[ \cap_{w \in (g \setminus \{u \}) \cap V_i } F_w^{\gamma} \mid c \in L_u(i+1), F_u^{\gamma},  \overline{B_c}  ] 	\right) \label{work_work_work} \\
 & = 1 -  \prod_{\ell =1}^{k-1} \prod_{ g \in D_{i,\ell}(u,\gamma)}   \left( 1 - \Pr[ \cap_{w \in (g \setminus \{u \}) \cap V_i } F_w^{\gamma} \mid c \in L_u(i+1),  \overline{B_c}  ]  \right). \label{work_work}
\end{align}
Note that in deriving~\eqref{work_work_work} we use the fact the girth of the hypergraph is at least $5$ which, in particular, implies that any two hyperedges that contain $u$ do not have any other vertex in common.

To further bound~\eqref{work_work},  we  consider the probability that every vertex in $( g \setminus \{u \} ) \cap V_i$ is activated and 
assigned $\gamma$ in $\sigma_i$, conditional on that   $c \in L_u(i+1)$ and $\overline{B_c}$, for any fixed $\ell \in [k-1]$ and $g \in D_{i,\ell}(u,\gamma)$.   We consider two cases depending on whether $g = h$ or not.

We start with the case where $g \ne h$. Let $A_g$ be the event that not every vertex in $(g \setminus \{u \} ) \cap V_{i}$  is activated and assigned $c$ in $\sigma_i$. Since our hypergraph has girth at least $5$ and the color activations and color assignments are independent over different vertices, we have:
\begin{align}
 \Pr[ \cap_{w \in (g \setminus \{u \}) \cap V_i } F_w^{\gamma} \mid c \in L_u(i+1),  \overline{B_c}  ]  & =  \Pr[ \cap_{w \in (g \setminus \{u \}) \cap V_i } F_w^{\gamma} \mid A_g ]   \nonumber  \\
 																		    & =  \frac{ \Pr[ ( \cap_{w \in (g \setminus \{u \}) \cap V_i } F_w^{\gamma}) \wedge A_g ] }{  \Pr[ A_g]}  \nonumber \\
																		    & = \frac{ \Pr[  \cap_{w \in (g \setminus \{u \}) \cap V_i } F_w^{\gamma} ] }{  \Pr[ A_g]}   \nonumber \\
																		    & = \frac{ \alpha^{\ell} L_i^{-\ell} }{ 1 - \alpha^{\ell} L_i^{-\ell}  }  \le  \left( \frac{\alpha}{L_i } \right)^{\ell} +   \frac{1 }{ L_i^{2\ell} (\ln \Delta)^{2\ell}  }, \label{stilll}
\end{align}
for sufficiently large $\Delta$, since $K<1$. 

Next we consider the case $g =h $. The difference here is that the event $\cap_{w \in (h \setminus \{u \}) \cap V_i } F_w^{\gamma} $ is not independent  of $\overline{B_c}$ as before. However, notice that since $g \in D_{i,\ell}(u,\gamma)$, $h \in D_{i,j}(v,c)$ and $\gamma \ne c$, we can only have $g= h$ when $j= \ell = k-1$.  This means that the occurrence of  event $\overline{B_c}$ prohibits the occurrence of event  $A_h$. Therefore, the event $\cap_{w \in (h \setminus \{u \}) \cap V_i } F_w^{\gamma}$  is independent of the event $c \in L_u(i+1)$ conditional on the event $\overline{B_c}$  and, thus, we have:
\begin{align}
 \Pr[ \cap_{w \in (h \setminus \{u \}) \cap V_i } F_w^{\gamma} \mid c \in L_u(i+1),  \overline{B_c}  ]  & =  \Pr[ \cap_{w \in (h \setminus \{u \}) \cap V_i } F_w^{\gamma} \mid \overline{B_c}  ]    \nonumber\\
																		    &=  \frac{ \Pr[ \cap_{w \in (h \setminus \{u \}) \cap V_i } F_w^{\gamma} ]  }{ \Pr[ \overline{B_c}  ]  }  \nonumber \\
																		    & =  \frac{ \alpha^{k-1} L_i^{-(k-1)} }{1 -   L_i^{-(k-1)}  }  \le 2 \left( \frac{\alpha}{L_i } \right)^{k-1}  \label{job_done},
\end{align}
for sufficiently large $\Delta$.

Combining~\eqref{work_work},~\eqref{stilll} and~\eqref{job_done}, we are able to show the following proposition.
\begin{proposition}\label{propopain} 
For every color $\gamma \ne c$:
\begin{align}\label{street}
\Pr[ \gamma \not \in L_u(i+1) \mid c \in L_u(i+1), F_u^{\gamma}, \overline{B_c}  ]  \le  1- \left(1 - 2 \left(  \frac{ \alpha}{L_i } \right)^{k-1} \right)   \mathrm{Keep}_i + 2 \sum_{\ell=1 }^{k-1} \frac{T_{i,\ell} }{L_i^{2 \ell } (\ln \Delta)^{2 \ell}}.
\end{align}
\end{proposition}
\begin{proof}
 Towards proving~\eqref{street}, a helpful observation is the following.
\begin{align}
\prod_{\ell= 1}^{ k-1}   \left( 1- \left( \frac{ \alpha }{L_i }\right)^{\ell } - \frac{ 1}{ L_i^{2 \ell} (\ln \Delta)^{2 \ell}  }   \right)^{T_{i,\ell} (u,\gamma) }   &= \prod_{\ell=1}^{k-1} \left( 1-  \left(\frac{\alpha}{L_i}\right)^{\ell} \right)^{T_{i,\ell}(u,\gamma) } \nonumber \\
																										&	\times  \left(1 -  \frac{1 }{L_i^{2\ell}  (\ln \Delta)^{2\ell}  \cdot ( 1-  ( \frac{  \alpha }{ L_i} )^{\ell} ) }  \right)^{T_{i,\ell}(u,\gamma) }  \nonumber \\
& \ge  \mathrm{Keep}_i \cdot  \mathrm{exp} \left( -  \sum_{ \ell = 1}^{k-1}  \frac{ T_{i,\ell} (u,\gamma) }{   L_i^{2\ell}  (\ln \Delta)^{2\ell}  \cdot ( 1-  ( \frac{  \alpha }{ L_i} )^{\ell} )    -1   }  	\right)  \label{pame_ligo_talep} \\
& \ge   \mathrm{Keep}_i \left( 1 -    \sum_{\ell=1 }^{k-1} \frac{T_{i,\ell} }{L_i^{2 \ell } (\ln \Delta)^{2 \ell}    \cdot ( 1-  ( \frac{  \alpha }{ L_i} )^{\ell} )    -1    } \right)  \nonumber \\
& \ge   \mathrm{Keep}_i \left( 1 -  \frac{3}{2} \cdot \sum_{\ell=1 }^{k-1} \frac{T_{i,\ell} }{L_i^{2 \ell } (\ln \Delta)^{2 \ell}     -  L_i^{\ell}  (\ln \Delta)^{\ell}       } \right)   \nonumber  \\
& \ge   \mathrm{Keep}_i  -  \frac{19}{10} \sum_{\ell=1 }^{k-1} \frac{T_{i,\ell} }{L_i^{2 \ell } (\ln \Delta)^{2 \ell}       } \label{I_like}.
\end{align}
for sufficiently large $\Delta$. Note that in~\eqref{pame_ligo_talep} we used the fact that $1- \frac{1}{x} \ge \mathrm{e}^{-\frac{1}{x-1} } $ for any $x \ge 2$.

Using~\eqref{stilll},~\eqref{job_done} and~\eqref{I_like}, we have:
\begin{align}
&\prod_{\ell =1}^{k-1} \prod_{ g \in D_{i,\ell}(u,\gamma)}   \left( 1 - \Pr[ \cap_{w \in (g \setminus \{u \}) \cap V_i } F_w^{\gamma} \mid c \in L_u(i+1),  \overline{B_c}  ]  \right) \ge   \nonumber \\
& \ge  \prod_{\ell= 1}^{ k-1}   \left( 1- \left( \frac{ \alpha }{L_i }\right)^{\ell } - \frac{ 1}{ L_i^{2 \ell} (\ln \Delta)^{2 \ell}  }   \right)^{T_{i,\ell} (u,\gamma) }  \left(  1- 2 \left( \frac{ \alpha }{L_i } \right)^{ k-1}   \right)          \label{sweet_home} \\
& \ge  \left( \mathrm{Keep}_i  -  \frac{19}{10} \sum_{\ell=1 }^{k-1} \frac{T_{i,\ell} }{L_i^{2 \ell } (\ln \Delta)^{2 \ell}       } \right)  \left(  1- 2 \left( \frac{ \alpha }{L_i } \right)^{ k-1}  \right) \label{lucky_to_be} \\
& = \mathrm{Keep}_i \left(1 - 2 \left(  \frac{ \alpha}{L_i } \right)^{k-1} \right)  - \frac{19}{10} \sum_{\ell=1 }^{k-1} \frac{T_{i,\ell} }{L_i^{2 \ell } (\ln \Delta)^{2 \ell}       } \left(1 - 2 \left(  \frac{ \alpha}{L_i } \right)^{k-1} \right)  \\
& \ge  \mathrm{Keep}_i \left(1 - 2 \left(  \frac{ \alpha}{L_i } \right)^{k-1} \right)   - 2 \sum_{\ell=1 }^{k-1} \frac{T_{i,\ell} }{L_i^{2 \ell } (\ln \Delta)^{2 \ell}}  ,  \label{greatest_city}
\end{align}
for sufficiently large $\Delta$.  Note that in deriving~\eqref{sweet_home} we used our previous observation that~\eqref{job_done} applies only when $h = g$, and this can potentially happen only when $\ell = j = k-1$. 

Combining~\eqref{work_work}  with~\eqref{greatest_city} concludes the proof.

\end{proof}

Overall, combining~\eqref{first_casara},~\eqref{just_you_wait} and Proposition~\ref{propopain}, recalling that $S_i = \sum_{\ell=1 }^{k-1} \frac{T_{i,\ell} }{L_i^{2 \ell } (\ln \Delta)^{2 \ell} } $ and letting $\psi = 1 - 2 \left(  \frac{ \alpha}{L_i } \right)^{k-1}  $,  we see that $\Pr[E_{u,1} \mid \overline{B_c}] $ is at most
\begin{align*}
 &(1- \alpha )\mathrm{Keep}_i + 2   (L_i \ln \Delta)^{-j}   + \alpha  \frac{L_i -1 }{ L_i } ( \mathrm{Keep}_i + 2(L_i  \ln \Delta) ^{-j} ) \left(1 - \psi   \mathrm{Keep}_i   +   2S_i  \right)  \\
 &\le (1- \alpha )\mathrm{Keep}_i + 2   (L_i \ln \Delta)^{-j}  + ( \alpha \mathrm{Keep}_i +  2 \alpha (L_i  \ln \Delta) ^{-j}) \left(1 - \psi  \mathrm{Keep}_i   +   2S_i  \right)  \\
  &= (1- \alpha )\mathrm{Keep}_i +  \alpha \mathrm{Keep}_i   - \alpha \psi \mathrm{Keep}_i^2 + 2 \alpha \mathrm{Keep}_i S_i  + 2   (L_i \ln \Delta)^{-j}  \left( 1 + \alpha  \left(1 - \psi  \mathrm{Keep}_i   +   2S_i  \right)  \right)  \\
  & = \mathrm{Keep}_i( 1- \alpha \mathrm{Keep}_i)  +  \frac{2\alpha^{k}}{L_i^{k-1} } \cdot \mathrm{Keep}_i^2 +    2 a \mathrm{Keep}_i S_i +  2(L_i \ln \Delta)^{-j}  \left(  1    + \alpha \left( 1 - \psi \mathrm{Keep}_i   +   2S_i  \right) \right)  \\
  & \le \mathrm{Keep}_i( 1- \alpha \mathrm{Keep}_i)  + 2 \alpha S_i +  \frac{2\alpha^{k}}{L_i^{k-1} } +  2(L_i \ln \Delta)^{-j}  \left(  1    + \alpha \left( 1  +   2S_i  \right) \right)  \\
  & = \mathrm{Keep}_i( 1- \alpha \mathrm{Keep}_i)  + 2 \alpha S_i +  \frac{2\alpha^{k}}{L_i^{k-1} } +  (L_i \ln \Delta)^{-j}  \left(  2    + 2\alpha  +   4\alpha S_i   \right)  \\
  & = \mathrm{Keep}_i( 1- \alpha \mathrm{Keep}_i)  + 2 \alpha S_i  +  (L_i \ln \Delta)^{-j}  \left(  2    + 2\alpha  +    4\alpha S_i +   (\ln \Delta)^{j} \cdot \frac{2\alpha^{k}}{L_i^{k-1 -j} }  \right)  \\
  & = \mathrm{Keep}_i( 1- \alpha \mathrm{Keep}_i)  + 2 \alpha S_i  +  (L_i \ln \Delta)^{-j}  \left(  2   +    4\alpha S_i +   O\left(\frac{1}{\ln \Delta} \right)  \right)  \\
& \le\, \mathrm{Keep}_i( 1- \alpha \mathrm{Keep}_i)  +  \left( 2\alpha S_i +  (L_i \ln \Delta)^{-j} ( 3 + 4 \alpha S_i ) \right).
\end{align*}
for sufficiently large $\Delta$,  since $\mathrm{Keep}_i  < 1$ and $j \le k-1$.
This concludes the proof of~\eqref{E1_bound}.

Finally, combining~\eqref{three_two_second},~\eqref{E2_bound} and ~\eqref{E1_bound} we obtain
\begin{eqnarray}
\Pr[ h  \in D_{i+1,r}(v,c) ] &\le &     { j \choose  r } q_1^r   q_2  ^{j-r} \left( 1 + \delta_1 q_1^{-1} \right)^{r } \left( 1 + \delta_2 q_2^{-1}\right)^{j-r}     + \frac{2^k}{ L_i^j} \nonumber \\
				      & \le & { j \choose r} q_1^r q_2^{j-r} \left( 1 + 2r \delta_1 q_1^{-1} \right) \left(  1 +  2 (j-r) \delta_2 q_2^{-1} \right) + \frac{2^k}{ L_i^j}  \label{number_one} \\
				      & \le& { j\choose r} q_1^r q_2^{j-r}   + 		2 r {j \choose r} q_1^{ r-1} q_2^{j-r} \delta_1  			 +O\left(  \frac{1}{ L_i^j } \right) \label{number_two} \\
				      &\le & { j\choose r} q_1^r q_2^{j-r}   + 		4 r {j \choose r} \frac{ \mathrm{Keep}_i^{j-1} \alpha^{ j-r+1} S_i }{ L_i^{j-r} } 			 +O\left(  \frac{1}{ L_i ^j } \right),
\end{eqnarray}
concluding the proof of~\eqref{tornado_warning}, which was our goal. Note that in~\eqref{number_one} we used that $\mathrm{Keep}_i$ is bounded below by a constant according to   Corollary~\ref{bounding_keep} (the lower bound  only requires the assumptions of Lemma~\ref{bounding_keep_lemma}) and that $\delta_1 q_1^{-1} $, 
$\delta_2 q_2^{-1} $ are sufficiently small for large enough $\Delta$. In~\eqref{number_two} we used the fact that $\delta_2 q_2^{-1} = O(1/L_i^j)$ for $j \ge 2$. (We only care about $j \ge 2$ since if $j= 1$ then $r =1$ as well and, therefore, $j-r = 0$.)
\end{proof}

\subsection{Proof of Lemma~\ref{concentration_lemma}}

Let $\mathrm{Bin}(n,p)$ denote the binomial random variable that counts the number of successes in $n$ Bernoulli trials,  where each trial succeeds with probability $p$. We will find the following lemma useful (see, e.g., Exercise 2.12 in~\cite{mike_book}) :

\begin{lemma}\label{binomial_bound} 
For any $\gamma, \kappa ,n > 0$ we have 
\begin{align*}
\Pr\left[ \mathrm{Bin}\left(n,\frac{\gamma}{n} \right) \ge \kappa  \right] \le \frac{\gamma^{\kappa} }{ \kappa! }.  
\end{align*}
\end{lemma}

\begin{proof}[Proof of Part~\eqref{con_part_a}]

We will use Theorem~\ref{talagrand_inequality} to show that that the number of colors,  $\overline{\ell_v}$, which are removed from $L_v$ during iteration $i$ is highly concentrated. 

To that end, at first notice that our assumption that property $Q(i)$ holds and  part (a) of Lemma~\ref{expectations_lemma} imply
\begin{align*}
\ex[ \overline{\ell_v}] =  (1- \mathrm{Keep}_i) L_i \ge   (\ln \Delta)^{19(k-1)},
\end{align*}
for sufficiently large $\Delta$. The lower bound follows because we have assumed that $L_i \ge (\ln \Delta)^{20(k-1)}$ and, according to Corollary~\ref{bounding_keep}, $\mathrm{Keep}_i = \Omega(1)$.

Note now that changing the assignment (color or activation bit in $\sigma_i$) to any neighboring vertex of $v$ can change $\overline{\ell_v}$ by at most $1$, and changing the assignment to any other vertex cannot affect $\overline{\ell_v}$ at all.  Further, if $\overline{\ell_v} \ge s $, there are at most  $s$ groups of at most $k-1$ neighbors of $v$,  so that each vertex in each group received the same color, and each group corresponds to a different color from $L_v$. Thus, the color assignments and activation choices of these vertices certify that $\overline{\ell_v} \ge s$.

Given the above, we apply Theorem~\ref{talagrand_inequality} with $t = \ex[ \overline{\ell_v} ] ^{ \frac{1.9}{3 }  }$, $w = 2k$, $\gamma=1$,  to obtain
\begin{align*}
\Pr\left[ | \overline{\ell_v}  - \ex[  \overline{ \ell_v }]   | > L_i^{2/3}  \right] \le \Pr\left[ | \overline{\ell_v}  - \ex[  \overline{ \ell_v }]   | > t + 60 \gamma \sqrt{w \ex[ \overline{\ell(v)}]  }  \right]   \le 4\mathrm{e}^{- \frac{ (\ex[ \overline{\ell_v} ])^{ 0.8 /3 } }{ 8 \gamma^2 w} }  \le \Delta^{-\ln \Delta},
\end{align*}
for sufficiently large $\Delta$. 

Finally,  the fact that $\ex[ \ell_{i+1}(v)  ] = \ell_i(v) - \ex[ \overline{\ell_v}  ]  $ implies that
\begin{align*}
\Pr\left[ | \ell_{i+1}(v) - \ex[  \ell_{i+1}(v)]  | > L_i^{2/3}  \right]   = \Pr\left[ | \overline{ \ell_v} - \ex[ \overline{ \ell_v}  ]   | > L_i^{2/3}  \right] < \Delta^{ - \ln \Delta}. 
\end{align*}

\end{proof}

\begin{proof}[Proof of Part~\eqref{con_part_b}]
For the purposes of our analysis, we start by expressing $t_{i+1,r}'(v,c)$ as a difference of two random variables, $X_{i+1,r}(v,c)$ and $Y_{i+1,r}(v,c)$.

Recall the definition of $D_{i,r}(v,c)$ and let $Z_{i,r}(v,c) = \bigcup_{j = r}^{k-1} D_{i,j}(v,c)$. Let $X_{i+1,r}(v,c)$ denote the number of hyperedges  in $Z_{i,r}(v,c)$ which, in the end of the $i$-th iteration: (i) contain exactly $r$ uncolored vertices  other than $v$; and (ii) the rest of their vertices (excluding $v$) are colored  $c$ (either because they were already  colored from previous iterations, or because they were assigned color $c$ during iteration $i$ and retained it).  Define also $Y_{i+1,r}(v,c)$ as the number of these hyperedges containing an uncolored vertex $u \neq v$ such that: (i) $c$ is unavailable to $u$ at step $i+1$ ($c \notin L_u(i+1)$); (ii) the unavailability of $c$ for $u$ is independent of the color of $v$ in $\sigma_i$. In essence, $v$ is not the ``critical'' vertex responsible for $c \notin L_u(i+1)$. This ensures that $Y_{i+1,r}(v,c)$ only counts hyperedges where other vertices cause the color conflict, thereby keeping it distinct from the count for $t_{i+1,r}'(v,c)$.

%

By definition we have $t_{i+1,r}'(v,c) = X_{i+1,r}(v,c) - Y_{i+1,r}(v,c)$. Therefore, by the linearity of expectation, it suffices to show that $X_{i+1,r}(v,c)$ and $Y_{i+1,r}(v,c)$ are both sufficiently concentrated. This is because
\begin{align*}
& \Pr\left[  t_{i+1,r}'(v,c) - \ex[t_{i+1}'(v,c)]  > \frac{1}{2}\left( \sum_{j=r}^{k-1}  { j \choose r } \alpha^{j-r} \frac{T_{i,j}}{L_i^{j-r} } \right)^{2/3} \right],    \nonumber \\
=& \Pr\left[  X_{i+1,r}(v,c) - \ex[X_{i+1}(v,c)]   - \left( Y_{i+1,r}(v,c) - \ex[ Y_{i+1,r}(v,c)  ]    \right)  > \frac{1}{2} \left( \sum_{j=r}^{k-1}  { j \choose r } \alpha^{j-r} \frac{T_{i,j}}{L_i^{j-r} } \right)^{2/3}  \right],
\end{align*}
and, therefore, it is sufficient to prove that
\begin{eqnarray}
 \Pr\left[     X_{i+1,r}(v,c)   - \ex[X_{i+1,r}(v,c)]    >   \frac{1}{4} \left( \sum_{j=r}^{k-1}  { j \choose r }  \alpha^{j-r} \frac{T_{i,j}}{L_i^{j-r} } \right)^{2/3} \right] & \le & \frac{1}{ 2} \Delta^{- \ln \Delta },  \label{shuttle}    \\
  \Pr\left[ Y_{i+1,r}(v,c) - \ex[ Y_{i+1,r}(v,c)  ]    <  -   \frac{1}{4} \left( \sum_{j=r}^{k-1}  { j \choose r } \alpha^{j-r} \frac{T_{i,j}}{L_i^{j-r} } \right)^{2/3}    \right]  &\le &  \frac{1}{ 2} \Delta^{- \ln \Delta } \label{shuttle2}.
\end{eqnarray}

\paragraph{Proof of~\eqref{shuttle}.}
We start by observing that the value of $X_{i+1,r}(v,c)$ is highly sensitive to the color assignment of $v$ in $\sigma_i$; for instance, it may drop to zero if all uncolored vertices in they hyperedges of $Z_{i,r}(v,c)$ are activated and assigned color $c$ in $\sigma_i$, $v$ is also activated in $\sigma_i$, and we change the color of $v$ to $c$. This sensitivity precludes a direct application of Talagrand’s inequality. Consequently, our strategy is to isolate the specific influence of $v$'s color.

To that end, let $W_{i+1,r}^1(v,c)$ denote the number of hyperedges contributing to $X_{i+1,r}(v,c)$ such that every uncolored vertex $u \neq v$ satisfies one of two conditions: (i) $u$ was deactivated in $\sigma_i$; (ii) $u$ remains uncolored due to the assignments (colors and activation bits in $\sigma_i$) of vertices in hyperedges not containing $v$. By definition, modifying the color assignment to $v$ does not affect the value of $W_{i+1,r}^1(v,c)$. Let also $W_{i+1,r}^2(v,c)$ be the number of hyperedges where exactly $r$ vertices (excluding $v$) are activated and assigned the same color as $v$ in $\sigma_i$, while the remaining $k-1-r$ vertices (excluding $v$) are either activated and assigned color $c$ in $\sigma_i$, or were already colored $c$ in a previous iteration. It follows that $X_{i+1,r}(v,c) \le W_{i+1,r}^1(v,c) + W_{i+1,r}^2(v,c)$, as $W_{i+1,r}^2(v,c)$ accounts for all hyperedges in $X_{i+1,r}(v,c)$ not already covered by $W_{i+1,r}^1(v,c)$, potentially including additional hyperedges. (For example, $W_{i+1,r}^2(v,c)$ captures the hyperedges in $D_{i,r}(v,c)$ that were counted by $X_{i+1,r}(v,c)$ because all $r$ of their uncolored vertices (other than $v$) were activated and assigned the same color as $v$ under $\sigma_i$.)  Moreover, as we will see, $W_{i+1,r}^1(v,c)$ and $W_{i+1,r}^2(v,c)$ are amenable to concentration arguments.

First, we consider $W_{i+1,r}^1(v,c)$.  Since the hypergraph is free of 3-cycles, changing the color or activation bit in $\sigma_i$ for a vertex in hyperedge $h \in Z_{i,r}(v,c)$ only influences the uncolored status of vertices within $h$. Specifically, such a change cannot affect any other hyperedge in $Z_{i,r}(v,c)$; consequently, $W_{i+1,r}^1(v,c)$ changes by at most 1. Since the hypergraph also lacks 4-cycles, changing the color or activation of any vertex not in a hyperedge in $Z_{i,r}(v,c)$ can affect only one vertex in at most one hyperedge of $Z_{i,r}(v,c)$. Therefore, $W_{i+1,r}^1(v,c)$ remains stable within a margin of 1.
 
We claim now that if $W_{i+1,r}^{1}(v,c) \ge s$, then there exist at most $ 2k^2 s$ random choices that certify this event. To see this, consider a hyperedge $h$ counted by $W_{i+1,r}^1(v,c)$. For every vertex $u \in h \setminus \{v\}$ that failed to retain its color, one of the following must have occurred in $\sigma_i$: $u$ was deactivated; or $u$ is contained in a hyperedge $h' \ne h$ such that all vertices in $(h' \setminus \{u\}) \cap V_{i}$ were activated and received the same color as $u$. Observe also that the event where a vertex $u \in h \setminus \{v \}$ is activated and assigned color $c$ in $\sigma_i$ is determined by the outcome of two random choices. Consequently, we can certify that $h$ contributes to $W_{i+1,r}^1(v,c)$ using at most $2k^2$ random choices.

Observe now that 
\begin{align*}
 (\ln \Delta)^{19(k-1)}  \le  ( 1- \alpha)^r T_{i,r} \le \ex[W_{i+1,r}^{1}(v,c) ] \le \sum_{j=r}^{k-1}  { j \choose r }  \alpha^{j-r} \frac{T_{i,j}}{L_i^{j-r} },
\end{align*}
for sufficiently large $\Delta$, since  $T_{i,r } \ge (\ln \Delta)^{20(k-1) }$ according to the hypothesis of Lemma~\ref{concentration_lemma}.  Let  also 
\begin{align*}
\Delta_{W^1} :=  \left| W_{i+1,r}^{1}(v,c) -  \ex[ W_{i+1,r}^{1}(v,c) \right|.
\end{align*}
 Applying Theorem~\ref{talagrand_inequality} with $\gamma =1 $, $w = 2k^2$ and $t = \left(\ex[W_{i+1,r}^1(v,c)]\right)^{1.9/3}$, and for sufficiently large $\Delta$,  we obtain  
\begin{align}
\Pr\left[  \Delta_{W^1} > \frac{1}{8} \left( \sum_{j=r}^{k-1}  { j \choose r } \alpha^{j-r} \frac{T_{i,j}}{L_i^{j-r} } \right)^{2/3} \right] & \le \Pr\left[  \Delta_{W^1} > t + 60 \gamma \sqrt{w \ex[ W_{i+1,r}^1(v,c) ]  }  \right]  \nonumber \\
																														         & \le  4 \mathrm{e}^{ -  \frac{ (\ex[W_{i+1,r}(v,c) ])^{0.8/3 } }{8 \gamma^2 w }  }  \le \frac{ 1}{ 4} \Delta^{ -\ln  \Delta } \label{X1}.
\end{align}

As far as $W_{i+1,r}^2(v,c)$ is concerned, note that it is distributed as $\sum_{j = r}^{k-1} \mathrm{Bin}( T_{i,j}  ,  \frac{ \alpha^j} {L_i^j} )$, since Lemma~\ref{concentration_lemma} assumes that property $Q(i)$ holds (and therefore $|D_{i,j}(v,c)| = T_{i,j}$ fore every $j$ and $\ell_i(u) = L_i$ for every vertex $u$).  Recalling Lemma~\ref{bounding_keep_lemma}, we see that for every $j \ge r$
\begin{align*}
  \frac{ \alpha^j} {L_i^j} = \alpha^{j} R_{i,j}  \cdot \frac{1}{ T_{i,j}} \le       \alpha^{j} k^{2(k-1-j) }  \ln \Delta \cdot \frac{1}{T_{i,j} } \le \frac{\lambda}{ T_{i,j} } ,
\end{align*}
for some constant $\lambda$. Therefore, applying Lemma~\ref{binomial_bound}  with  $\kappa := \left\lfloor \frac{1}{16(k-1)} \left( \sum_{j=r}^{k-1} { j \choose r } \alpha^{j-r} \frac{T_{i,j}}{L_i^{j-r}} \right)^{2/3 }  \right\rfloor   $  we get that
\begin{align}\label{conc_talaip}
\Pr\left[  \mathrm{Bin}\left(T_{i,j}, \frac{ \alpha^j}{ L_i^j}  \right) \ge    \kappa    \right] & = \Pr\left[  \mathrm{Bin}\left(T_{i,j}, \frac{ \alpha^j R_{i,j} }{ T_{i,j}}  \right) \ge  \kappa \ \right]  \le  \frac{ \lambda^ {   \kappa  } }{   \kappa  ! }  \le \frac{1 }{k-1}  \cdot \frac{1}{4} \Delta^{- \ln \Delta },														
\end{align}
for sufficiently large $\Delta$, since  $\kappa =  \Omega(T_{i,r}^{2/3} ) =  \Omega\left( (\ln \Delta)^{\frac{40}{3}(k-1)}  \right)$, according to the hypothesis of Lemma~\ref{concentration_lemma}.
Thus, letting $E_j$ denote the event that  $\mathrm{Bin}\left(T_{i,j}, \frac{ \alpha^j}{ L_i^j}  \right) \ge    \kappa$, we obtain:
\begin{align}
&\Pr \left[   W_{i+1,r}^2 (v,c)  \ge \frac{1 }{8}     \left(\sum_{j=r}^{k-1} { j \choose r } \alpha^{j-r} \frac{T_{i,j}}{L_i^{j-r}} \right)^{2/3} \right]  \nonumber  \\
=& \Pr\left[ \sum_{j=r}^{k-1}  \mathrm{Bin} \left(T_{i,j}, \frac{ \alpha^j}{ L_i^j}  \right) \ge   \frac{1 }{4}   \left(   \sum_{j=r}^{k-1} { j \choose r } \alpha^{j-r} \frac{T_{i,j}}{L_i^{j-r}} \right)^{2/3}   \right]  \nonumber \\
																							 <& \Pr\left[ \bigcup_{j=r}^{k-1}   E_j  \right]  \le   \sum_{j=r}^{k-1} \Pr[ E_j]  \le \frac{1}{4} \Delta^{ - \ln \Delta} \label{X2}.
\end{align}
Note that for the strict inequality we used the fact that if none of the events $E_j$ occur then, by definition, $\sum_{j=r}^{k-1} \mathrm{Bin} \left(T_{i,j}, \frac{ \alpha^j}{ L_i^j}  \right) \le (k-1) \kappa < \frac{1 }{8}   \left(   \sum_{j=r}^{k-1} { j \choose r } \alpha^{j-r} \frac{T_{i,j}}{L_i^{j-r}} \right)^{2/3}  $.

Finally, we claim that the fact that $ W_{i+1}^1(v,c) \le X_{i+1,r}(v,c) \le W_{i+1}^1(v,c) + W_{i+1}^2(v,c) $ together with~\eqref{X1} and~\eqref{X2} implies~\eqref{shuttle}. To see this, let $\beta :=  \frac{1}{4} \left( \sum_{j=r}^{k-1}  { j \choose r }  \alpha^{j-r} \frac{T_{i,j}}{L_i^{j-r} } \right)^{2/3} $, and notice that
\begin{align}
&\Pr\left[ X_{i+1,r}(v,c) - \ex[ X_{i+1,r}(v,c)]  \ge  \beta  \right]   \nonumber \\
 \le & \Pr\left[ W_{i+1}^1(v,c) + W_{i+1}^2(v,c)  \ge \ex[ X_{i+1,r}(v,c)]  +   \beta  \right]  \nonumber  \\
											 \le & \Pr\left[ W_{i+1}^1(v,c) + W_{i+1}^2(v,c)  \ge \ex[ W_{i+1,r}^1(v,c)]  +   \beta  \right]   \nonumber \\
											 \le & \Pr[W_{i+1}^1(v,c)  \ge \ex[ W_{i+1,r}^1(v,c)]  + \beta/2  ] + \Pr[  W_{i+1}^2(v,c)  \ge \beta /2 ]  \nonumber \\
											\le & 2 \cdot \frac{1}{4} \Delta^{-\ln \Delta} = \frac{1 }{2} \Delta^{- \ln \Delta} \label{alm_there_conc_1},
\end{align}
where in the last inequality we used~\eqref{X1} and~\eqref{X2}, concluding the proof of~\eqref{shuttle}.
 
%

\paragraph{Proof of~\eqref{shuttle2}.} We apply a similar approach to analyze $Y_{i+1,r}(v,c)$. Let $U_{i+1,r}^1(v,c)$ denote the number of hyperedges contributing to $Y_{i+1,r}(v,c)$ in which every uncolored vertex $v' \ne v$ satisfies one of two conditions: (i) $v'$ was deactivated in $\sigma_i$, or (ii) $v'$ remains uncolored determined solely by the assignments (colors and activation bits) of vertices in hyperedges disjoint from $v$. Put differently, $U_{i+1,r}^1(v,c)$ counts the hyperedges in the intersection of the sets counted by $W_{i+1,r}^1(v,c)$ and $Y_{i+1,r}(v,c)$. Observe that the value of $U_{i+1,r}^1(v,c)$ is invariant under the color assignment of $v$. This holds because the definition of $Y_{i+1,r}(v,c)$ selects hyperedges $h$ that contain a vertex $u \in h \setminus \{v\}$  where the unavailability of $c$ at $u$ is independent of $v$'s color. Consequently, Talagrand’s inequality applies to $U_{i+1,r}^1(v,c)$ almost identically to its application for $W_{i+1,r}^1(v,c)$, as shown below.

Since the hypergraph is free of 3-cycles, changing the color or activation bit in $\sigma_i$ for a vertex in hyperedge $h \in Z_{i,r}(v,c)$ only influences the uncolored status and the candidate list of colors of vertices within $h$. Specifically, such a change cannot affect any other hyperedge in $Z_{i,r}(v,c)$; consequently, $U_{i+1,r}^1(v,c)$ changes by at most 1. Since the hypergraph also lacks 4-cycles, changing the color or activation of any vertex not in a hyperedge in $Z_{i,r}(v,c)$ can affect only one vertex in at most one hyperedge of $Z_{i,r}(v,c)$. Therefore, $U_{i+1,r}^1(v,c)$ remains stable within a margin of 1.
 
We claim now that if $U_{i+1,r}^{1}(v,c) \ge s$, then there exist at most $ (2k^2  + 2k)s$ random choices that certify this event. To see this, consider a hyperedge $h$ counted by $U_{i+1,r}^1(v,c)$. For every vertex $v' \in h \setminus \{v\}$ that failed to retain its color, one of the following must have occurred in $\sigma_i$: $v'$ was deactivated; or $v'$ is contained in a hyperedge $h' \ne h$ such that all vertices in $(h' \setminus \{v'\}) \cap V_{i}$ were activated and received the same color as $v'$.  Also, observe  that the event where a vertex $v' \in h \setminus \{v \}$ is activated and assigned color $c$ in $\sigma_i$ is determined by the outcome of two random choices. Finally, let $u$ be the vertex for which $c \notin L_u(i+1)$ independently of $v$'s color. This means $u$ is contained in a hyperedge $g$ such that all vertices in $(g \setminus \{u\}) \cap V_{i}$ were activated and assigned color $c$ in $\sigma_i$.  Consequently, we can certify that $h$ contributes to $U_{i+1,r}^1(v,c)$ using at most $2k^2 + 2k$ random choices.

To continue, we derive upper and lower bounds for the expectation $\mathbb{E}[U_{i+1,r}^1(v,c)]$. For the upper bound, observe that:
\begin{align}
            \ex[U_{i+1,r}^1(v,c) ]   \le \ex[ W_{i+1,r}^1(v,c)] \le  \sum_{j=r}^{k-1} {j  \choose r }  \alpha^{j-r}  \frac{ T_{i,j}}{L_i^{j-r} }. \label{upp_u}
\end{align}
For the lower bound, at first recall that property $Q(i)$ holds by the hypothesis of Lemma~\ref{concentration_lemma}. Consider now a hyperedge  $h \in D_{i,r}(v,c)$  and define
\begin{align*}
\widetilde{\mathrm{Keep}_i}:=  \prod_{j \in [k-1] \setminus \{r \}}  \left( 1 - \left(  \frac{ \alpha}{ L_i } \right)^j\right)^{T_{i,j}} \cdot  \left( 1 - \left( \frac{ \alpha }{L_i } \right)^r \right)^{T_{i,r} -1 }
\end{align*}
to be the probability that a vertex $u \in (h \setminus \{v\}) \cap V_i $ has $c$ available in the next iteration (i.e., $c \in L_u(i+1)$),  assuming the color assignments and activations of vertices in $h$—including $v$—are ignored. Notice then that
\begin{align}
 \ex [ U_{i+1,r}^1 (v,c) ]   \ge  (1 - \alpha)^r (1- \widetilde{\mathrm{Keep}_i} ) T_{i,r}, \label{lower_u}
\end{align}
since $(1 - \alpha)^r (1-\widetilde{\mathrm{Keep}_i})$ is a lower bound on the probability that a certain edge $h \in D_{i,r}(v,c)$ is counted by $U_{i+1,r}^1(v,c)$  (corresponding to the event that all $r$ vertices in $h \setminus\{v\} \cap V_i$ are de-activated in $\sigma_i$, $c$ is unavailable for a certain vertex $u \in h \setminus \{v \} \cap V_i$ in the end of iteration $i$, and the unavailability of $c$  for $u$ is independent of the color of $v$ in $\sigma_i$). 

Akin to the proof of Corollary~\ref{bounding_keep},  making use of the assumption that $R_{i,k-1} := \frac{ T_{i,k-1}}{ L_i^{k-1}} \ge \frac{1}{10k^2}$, and for sufficiently large $\Delta$, we obtain:
\begin{align*}
\widetilde{\mathrm{Keep}_i}  &   \le \left( 1 - \left( \frac{ \alpha }{L_i } \right)^{k-1} \right)^{T_{i,k-1} -1 } =   \frac{ \left( 1 - \left( \frac{ \alpha }{L_i } \right)^r \right)^{T_{i,k-1}  }  }{1- \frac{\alpha^{k-1}}{L_i^{k-1}} } \le \frac{ \mathrm{exp}(-  \alpha^{k-1}R_{i,k-1}  )} { 1- \frac{\alpha^{k-1}}{L_i^{k-1}} }  \\
 					& \le  \mathrm{exp}(-  \alpha^{k-1}R_{i,k-1}  ) \left( 1 + \frac{2 \alpha^{k-1}}{L_i^{k-1}}  \right) < \left( 1 -  \frac{K^{k-1}}{11 k^2  (\ln \Delta )^{k-1 } } \right) \left( 1 + \frac{2 \alpha^{k-1}}{L_i^{k-1}}  \right)  \\
					& \le 1 -  \frac{K^{k-1}}{12 k^2  (\ln \Delta )^{k-1 } }.
\end{align*}

Combining the above with~\eqref{lower_u} we get:
\begin{align}
 \ex [ U_{i+1,r}^1 (v,c) ]   \ge  (1 - \alpha)^r (1- \widetilde{\mathrm{Keep}_i} ) T_{i,r} \ge  (1 - \alpha)^r \cdot  \frac{K^{k-1}}{12 k^2  (\ln \Delta )^{k-1 } } \cdot (\ln \Delta)^{20(k-1)} \ge (\ln \Delta)^{18(k-1)}   \label{final_lower_u}.
\end{align}

Given~\eqref{upp_u} and~\eqref{final_lower_u},  we are now ready to apply Talagrand's inequality.  Let 
\begin{align*}
\Delta_{U^1} :=  \left| U_{i+1,r}^{1}(v,c) -  \ex[ U_{i+1,r}^{1}(v,c) \right|.
\end{align*}
Applying Theorem~\ref{talagrand_inequality} with $\gamma =1 $, $w = 2k^2 + 2k$ and $t = \left(\ex[U_{i+1,r}^1(v,c)]\right)^{1.9/3}$, and for sufficiently large $\Delta$,  we obtain  
\begin{align}
\Pr\left[  \Delta_{U^1} > \frac{1}{8} \left( \sum_{j=r}^{k-1}  { j \choose r } \alpha^{j-r} \frac{T_{i,j}}{L_i^{j-r} } \right)^{2/3} \right] & \le \Pr\left[  \Delta_{U^1} > t + 60 \gamma \sqrt{w \ex[ U_{i+1,r}^1(v,c) ]  }  \right]  \nonumber \\
																														         & \le  4 \mathrm{e}^{ -  \frac{ (\ex[U_{i+1,r}(v,c) ])^{0.8/3 } }{8 \gamma^2 w }  }  \le \frac{ 1}{ 2} \Delta^{ -\ln  \Delta } \label{UX1}.
\end{align}

Recall the definition of $W_{i+1,r}^2(v,c) $ and observe that
\begin{align}\label{papaoutai_afro_soul_1}
U_{i+1,r}^1(v,c) \le Y_{i+1,r}(v,c) \le U_{i+1,r}^1(v,c)  + W_{i+1,r}^2(v,c),
\end{align}
and that
\begin{align}\label{papaoutai_afro_soul_2}
 \ex[W_{i+1,r}^2(v,c) ] = \ex\left[ \sum_{j = r}^{k-1} \mathrm{Bin}\left( T_{i,j}  ,  \frac{ \alpha^j} {L_i^j} \right)  \right] =  \sum_{j=r}^{k-1} T_{i,j} \frac{\alpha^j }{L_i^j }   = \sum_{j=r}^{k-1} R_{i,j} \alpha^j = O(1).
 \end{align}
Recalling now that  $\beta :=  \frac{1}{4} \left( \sum_{j=r}^{k-1}  { j \choose r }  \alpha^{j-r} \frac{T_{i,j}}{L_i^{j-r} } \right)^{2/3} $ and combining~\eqref{UX1}, ~\eqref{papaoutai_afro_soul_1} and~\eqref{papaoutai_afro_soul_2}, we obtain:
\begin{align}
\Pr\left[ Y_{i+1,r}(v,c) - \ex[ Y_{i+1,r}(v,c)]  \le - \beta  \right]  &  \le \Pr\left[ U_{i+1,r}^1(v,c) -  \ex[Y_{i+1,r}(v,c)]  \le  - \beta  \right]  \nonumber \\
											   & \le \Pr\left[ U_{i+1,r}^1(v,c) -  \ex[ U_{i+1,r}^1(v,c) + W_{i+1,r}^2(v,c)]  \le  - \beta  \right]  \nonumber \\
											  & = \Pr\left[ U_{i+1,r}^1(v,c)  -  \ex[ U_{i+1,r}^1(v,c)] \le \ex[W_{i+1,r}^2(v,c) ] - \beta \right] \nonumber   \\
											  & \le \Pr\left[ U_{i+1,r}^1(v,c)  -  \ex[ U_{i+1,r}^1(v,c)] \le O(1) - \beta \right] \nonumber  \\
											  & \le \frac{1}{2} \Delta^{ - \ln \Delta}, \nonumber
\end{align}
for sufficiently large $\Delta$, concluding the proof of~\eqref{shuttle2}.

\end{proof}

\subsection{Proof of Lemma~\ref{final_corollary}}

We  use induction on $i$. Property $P(1)$ clearly holds, so we assume that property $P(i)$ holds and we prove that with property $P(i+1)$ holds with positive probability. Recall our discussion in the previous section in which we argued that we can assume without loss of generality that property $Q(i)$ holds. 

For every $v$ and $c \in L_v$ let $A_v$ be the event that $\ell_{i+1}(v) < L_{i+1}$ and $B_{v,c}^r$ to be the event that $t_{i+1,r}(v,c) > T_{i+1,r}$. Clearly, if these bad events are avoided, then $P(i+1)$ holds.

Since property $Q(i)$ holds, we have that $\ell_i(v) = L_i$. Therefore, by~\eqref{L_def} and Lemmas~\ref{expectations_lemma},~\ref{concentration_lemma}  we have:
\begin{align}\label{probA}
\Pr[A_v ] = \Pr\left[ \ell_{i+1}(v) < L_i \cdot \mathrm{Keep} - L_i^{2/3} \right]   = \Pr\left[ \ell_{i+1}(v) < \ex[ \ell_{i+1}(v) ] - L_i^{2/3} \right]  < \Delta^{- \ln \Delta }. 
\end{align} 
Similarly, by~\eqref{T_def} and Lemmas~\ref{expectations_lemma},~\ref{concentration_lemma} we have:
\begin{align}
\Pr[B_{v,c}^r]  &\le \Pr[ t_{i+1,r}'(v,c) > T_{i+1,r}  ] \nonumber  \\
                      & = \Pr\left[ t_{i+1,r}'(v,c) - \ex[t_{i+1,r}'(v,c) ]  >   \left( \sum_{j=r}^{k-1} { j \choose r} \alpha^{j-r} \frac{T_{i,j} }{L_i^{j-r } }    \right)^{2/3 }-   C' \cdot  \sum_{j=r}^{ k-1} \frac{T_{i,j} }{L_i^j }   \right]  \nonumber\\
                      & \le   \Pr\left[ t_{i+1,r}'(v,c) - \ex[t_{i+1,r}'(v,c) ]  >   \frac{1}{2} \left( \sum_{j=r}^{k-1} { j \choose r} \alpha^{j-r} \frac{T_{i,j} }{L_i^{j-r } }    \right)^{2/3 } \right]  \label{snoop}  \\
                      & \le  \Delta^{- \ln \Delta } \label{ProbB},
\end{align}
where $C' > 0$ is the hidden constant that multiplies $Y_i$ in the statement of Lemma~\ref{expectations_lemma}, Part (b). Note that in deriving~\eqref{snoop} we used Lemma~\ref{bounding_keep_lemma} — which implies that  $\sum_{j=r}^{ k-1} \frac{T_{i,j} }{L_i^j}  = O(\ln \Delta)$ — and our assumptions that  $L_{i} \ge  (\ln \Delta)^{20(k-1) } $ and that  $T_{i,k-1} \ge \frac{L_i^{k-1} }{10k^2} $ —  which imply that 
\begin{align*}
 \left( \sum_{j=r}^{k-1} { j \choose r} \alpha^{j-r} \frac{T_{i,j} }{L_i^{j-r } }    \right)^{2/3 } >  \left(  {k-1 \choose  r }  \alpha^{k-1-r}  \frac{T_{i,k-1} }{ L_j^{k-1-r } } \right)^{2/3}  \ge  \left( \alpha^{k-1-r} \frac{L_i^r}{10k^2 }  \right)^{2/3} = \omega (\ln \Delta).
 \end{align*}

Notice now that each bad event $f_v \in \{ A_v, B_{v,c}^r \}$ event is determined by the colors assigned to vertices of distance at most $3$ from $v$. Therefore, $f_v$ is mutually independent of all but at most $ (k\Delta)^4  (1+\delta) \left( \frac{ \Delta }{ \ln \Delta} \right)^{ \frac{1}{k-1} }   < \Delta^5$  other bad events. For $\Delta$ sufficiently large, $\Delta^{- \ln \Delta } \Delta^{5} < \frac{1}{4} $ and so the proof is concluded by applying Corollary~\ref{symmetric_LLL} using~\eqref{probA},~\eqref{ProbB}.

\subsection{Proof of Lemma~\ref{no_errors}}

\begin{proof}[Proof of part (a)]

Since $L_i < L_i'$, for the first part of the lemma it suffices to prove that $L_i' \le L_i + (L_i')^{5/6}$. Towards that end, at first we observe that for sufficiently large $\Delta$,  Corollary~\ref{bounding_keep}  and the fact that $K = \frac{1}{100 k^{3k} } $ imply:
\begin{eqnarray}
\mathrm{Keep}_i^{5/6} - \mathrm{Keep}_i  &\ge &  \left(1 -  \frac{K^{k-1}}{12 k^2  (\ln \Delta)^{k-1 }  } \right)^{5/6}  -  \left(1 -  \frac{K^{k-1}}{12 k^2  (\ln \Delta)^{k-1} }\right) \nonumber \\
								& \ge & \left( 1 - \frac{5}{6} \cdot \frac{ K^{k-1} }{ 12 k^2 (\ln \Delta)^{k-1} }  \right) - \left( 1 - \frac{ K^{k-1} }{ 12 k^2 (\ln \Delta)^{k-1 }  }  \right) = \frac{ K^{k-1} }{ 72 k^2 (\ln \Delta)^{k-1} }. \label{useful_obs} 
\end{eqnarray}
Note that in deriving the first inequality we used the fact that the function $x^{5/6} - x$ is decreasing on the interval $[C  ,1]$  since  $K  $ is sufficiently small. For the second one, we used the Taylor Series  for $(1-y)^{5/6}$ around $y = 0$.

We now proceed by using induction. The base case  is  trivial, so assume that the statement is true for $i$, and consider $i+1$. Since by our assumption $L_i \ge (\ln \Delta)^{20(k-1) } $ we obtain
\begin{eqnarray}
L_{i+1}' & =&  \mathrm{Keep}_i L_i'  \nonumber \\
	    & \le & \mathrm{Keep}_i \left( L_i + (L_i')^{5/6} \right)  \label{tsipras_woooow} \\
	    & = &L_{i+1} + L_i^{2/3} + \mathrm{Keep}_i (L_i')^{5/6} \nonumber \\
	    & \le & L_{i+1} +  L_i^{2/3} +  \mathrm{Keep}_i^{5/6} (L_i')^{5/6}-  \frac{K^{k-1}}{ 72 k^2 (\ln \Delta)^{k-1 } } (L_i')^{5/6} \label{tsipras_wow} \\
	    & \le& L_{i+1} +  	(L_{i+1}')^{5/6}		 +  L_i^{2/3} 	-  \frac{K^{k-1}}{ 72 k^2 (\ln \Delta)^{k-1 } } (L_i')^{5/6} \nonumber		\\
	    & < & L_{i+1} + (L_{i+1}')^{5/6}, \label{final_no_errors}
\end{eqnarray}
for sufficiently large $\Delta$. Note that in deriving~\eqref{tsipras_woooow} we used the inductive hypothesis; for~\eqref{tsipras_wow} we used~\eqref{useful_obs}; and for~\eqref{final_no_errors} the fact that $L_i \ge (\ln \Delta)^{20(k-1) }$ and the inductive hypothesis. 
\end{proof}

\begin{proof}[Proof of part (b)]
We observe that it suffices to show that $T_{i,r}' \ge T_{i,r} -  (T_{i,r}')^{\frac{100r}{100r +1}}    $  (since $T_{i,r} \ge T_{i,r}'$ for every $i$ by definition) and proceed by using induction. Again, the base case is trivial, so we assume the statement is true for $i$, and  consider $i+1$.

Let $\gamma_x = \frac{100x}{100x+1}$.  Recalling~\eqref{T_def} and the fact that $L_i < L_i'$ according to~\eqref{L_def} and~\eqref{Lprime_def}, and neglecting the positive lower-order additive noise terms which only the strengthen the inequality, we obtain:
\[
    T_{i+1,r} \ge \sum_{j=r}^{k-1} T_{i,j} {j \choose r} (\mathrm{Keep}_i(1-\alpha \mathrm{Keep}_i))^r \left( \frac{\alpha \mathrm{Keep}_i}{L_i'} \right)^{j-r}.
\]
Substituting the inductive hypothesis $T_{i,j} \ge T'_{i,j} - (T'_{i,j})^{\gamma_j}$ into the sum:
\begin{align}
    T_{i+1,r} &\ge \sum_{j=r}^{k-1} \left[ T'_{i,j} - (T'_{i,j})^{\gamma_j} \right] {j \choose r} (\mathrm{Keep}_i(1-\alpha \mathrm{Keep}_i))^r \left( \frac{\alpha \mathrm{Keep}_i}{L_i'} \right)^{j-r} \nonumber \\
    &= \sum_{j=r}^{k-1} T'_{i,j} {j \choose r}(\dots) - \sum_{j=r}^{k-1} (T'_{i,j})^{\gamma_j} {j \choose r}(\mathrm{Keep}_i(1-\alpha \mathrm{Keep}_i))^r \left( \frac{\alpha \mathrm{Keep}_i}{L_i'} \right)^{j-r}  \nonumber \\
    &=    T_{i+1,r}' - \sum_{j=r}^{k-1} (T'_{i,j})^{\gamma_j} {j \choose r}(\mathrm{Keep}_i(1-\alpha \mathrm{Keep}_i))^r \left( \frac{\alpha \mathrm{Keep}_i}{L_i'} \right)^{j-r}    \nonumber \\
    & = T_{i+1,r'} - \mathcal{E}_{total}   \label{eq:expansion}
\end{align}
We need to show that the total subtracted error $\mathcal{E}_{total}$ is bounded by the target slack $(T'_{i+1,r})^{\gamma_r}$. We split $\mathcal{E}_{total}$ into the dominant term ($j=r$) and the tail terms ($j > r$):
\[
    \mathcal{E}_{total} = \mathcal{E}_{dom} + \mathcal{E}_{tail}
\]
where
\begin{align*}
    \mathcal{E}_{dom} &= (T'_{i,r})^{\gamma_r} (\mathrm{Keep}(1-\alpha \mathrm{Keep}))^r, \\
    \mathcal{E}_{tail} &= \sum_{j=r+1}^{k-1} {j \choose r} (T'_{i,j})^{\gamma_j} (\mathrm{Keep}(1-\alpha \mathrm{Keep}))^r \left(\frac{\alpha \mathrm{Keep}}{L_i'}\right)^{j-r}.
\end{align*}

We start by bounding $\mathcal{E}_{\mathrm{tail}}$.  Lemma~\ref{pain} implies that there exists a constant $C_k$ such that, for every $i, j$, 
\begin{align}\label{easy_fact}
\frac{T_{i,j}' }{ (L_i')^j }  = R_{i,j}'  \le C_k \ln \Delta.
\end{align}
Moreover, since $T_{i, k-1} \ge \frac{ 1}{ 10k^2} L_i^{k-1}$ according to the hypothesis of Lemma~\ref{no_errors},  definition~\eqref{Lprime_def}, Corollary~\ref{bounding_keep}, and the first part of Lemma~\ref{no_errors}, imply that for every $j$:
\begin{align}
T_{i+1,j}'  \ge T_{i,k-1}' {k-1 \choose j} \left( \mathrm{Keep}_i (1 - \alpha \mathrm{Keep}_i )  \right)^r   \left(  \frac{ \alpha \mathrm{Keep}_i}{ L_i'}  \right)^{(k-1)-j} = \Omega \left(  \frac{(L_i')^{j} }{ (\ln \Delta)^{k-1-j} }  \right) \label{feeding_lower_bound}
\end{align}

Using~\eqref{easy_fact} we get:
\begin{align}
\mathcal{E}_{tail} &= \sum_{j=r+1}^{k-1} {j \choose r} (T'_{i,j})^{\gamma_j} (\mathrm{Keep}_i(1-\alpha \mathrm{Keep}_i))^r \left(\frac{\alpha \mathrm{Keep}_i}{L_i'}\right)^{j-r}  \nonumber \\
			 &= \sum_{j=r+1}^{k-1} {j \choose r} (R'_{i,j})^{\gamma_j} (\mathrm{Keep}_i(1-\alpha \mathrm{Keep}_i))^r \left(\alpha \mathrm{Keep}_i\right)^{j-r} \frac{1 }{(L_i')^{j-r  - j \gamma_j }  } \nonumber \\
			 & = O\left( \sum_{j=r+1}^{k-1}  \frac{ (L_i')^{r - j (1-\gamma_j) } }{  (\ln \Delta)^{j-r - \gamma_j } } \right)  =  O\left(   \frac{ (L_i')^{r - (r+1) (1-\gamma_{r+1}) } }{  (\ln \Delta)^{1- \gamma_j } }   \right)    \label{tail_error_bound}.
\end{align}
Now~\eqref{feeding_lower_bound} implies that:
\begin{align}\label{final_tail_error_bound}
(T_{i+1, r}')^{ \gamma_r}   =  \Omega \left(  \frac{(L_i')^{r \gamma_r} }{ (\ln \Delta)^{k-1-\gamma_r} }  \right)   = \omega( \mathcal{E}_{tail}), 
\end{align}
since $L_i' \ge L_i \ge ( \ln \Delta )^{20(k-1)}$ and $r \gamma_r -  \left( r - (r+1) (1-\gamma_{r+1}) \right) > 0$. To see the latter inequality, let
\begin{align*}
E(r)  = & r \gamma_r - (r - (r+1)(1 - \gamma_{r+1})) \\
 = & r(\gamma_r - 1) + (r+1)(1 - \gamma_{r+1}).
\end{align*}
Since  $\gamma_x = \frac{100x}{100x+1}$, we have that $1 - \gamma_x = \frac{1}{100x+1}$. Therefore,
\begin{align}
E(r) &= -\frac{r}{100r+1} + \frac{r+1}{100(r+1)+1} \nonumber \\
&= \frac{(r+1)(100r+1) - r(100r+101)}{(100r+1)(100r+101)} \nonumber  \\
&= \frac{(100r^2 + 101r + 1) - (100r^2 + 101r)}{(100r+1)(100r+101)} \nonumber \\
&= \frac{1}{(100r+1)(100r+101)} > 0 \quad \forall r \ge 0.
\end{align}

We continue by bounding $\mathcal{E}_{\mathrm{dom}}$.  Let $\Psi = \mathrm{Keep}_i(1-\alpha \mathrm{Keep}_i )$ and note that $\Psi <1$ for sufficiently large $\Delta$ according to Corollary~\ref{bounding_keep}. With this notation,
\begin{align}
T_{i+1,r}'  &\ge  T_{i,r}' \Psi^r,   \label{t_approx_psi}\\
\mathcal{E}_{dom} &= (T'_{i,r})^{\gamma_r} \Psi^r, \nonumber
\end{align}
where~\eqref{t_approx_psi} is implied by~\eqref{Lprime_def}. Thus,
\begin{align}\label{dom_error_term_bound}
\mathcal{E}_{dom} &= (T'_{i,r})^{\gamma_r} \Psi^r =  ((T'_{i,r})^{\gamma_r} \Psi^{\gamma_r} ) \Psi^{r - \gamma_r }  \le \Psi^{r - \gamma_r }  (T_{i+1,r}')^{ \gamma_r}.
\end{align}
Since $\Psi, \gamma_r < 1$ this  inequality implies a multiplicative gap between the dominant error and $(T_{i+1,r}')^{ \gamma_r}$, which allows us to absorb the asymptotically smaller tail terms $\mathcal{E}_{\mathrm{tail}}$ and show that $ \mathcal{E}_{\mathrm{total}} < (T_{i+1,r}')^{ \gamma_r}$, concluding the proof.

\end{proof}

\subsection{Proof of Lemma~\ref{target_lemma}}

We proceed by induction. Let $\eta := \frac{\epsilon/3 }{ (k-1)  (1 + \epsilon/2 ) } $.  We will assume that  $L_j \ge \Delta^{ \eta}  , T_{j,r} \ge (\ln \Delta)^{20(k-1) } $ for all $ 2 \le j \le i <i^*$,  and prove that  $L_{i+1} \ge \Delta^{ \eta}, T_{i+1,r} \ge (\ln \Delta)^{20(k-1) }$.
Towards that end, it will be useful to focus on the family of ratios $R_{i,r} $, $ r\in [k-1$].  Note that, according to Lemma~\ref{no_errors},  this family is well-approximated by the family  $R_{i,r}' $, $r \in [k-1]$.  In particular, recalling Lemma~\ref{pain} and  applying Lemma~\ref{no_errors} we obtain:
\begin{align}
R_{i,r}  & \le  R_{i,r}'  \cdot \frac { 1 + (T_{i,r}')^{-\frac{1}{100r+1}}   }{  \left(1-  (L_i')^{ -1/6 }  \right)^{r }  }  \nonumber \\
	& \le (1- \alpha C)^{r (i-1) } \ln \Delta \cdot  \frac{ \prod_{p=r}^{k-2 } (p+1)   } { (1+ \delta - \frac{1.1 \delta }{k^{99} })^{k-1 }  C^{k-1-r } }   ,\label{Gs}
\end{align} 
for sufficiently large $\Delta$, since $L_i,T_{i,r}  \ge (\ln \Delta)^{20(k-1)} $.

Using~\eqref{Gs} and the fact that $1 - \frac{1}{x} > \mathrm{e}^{- \frac{1}{x-1} } $ for $x \ge 2$  we can get an improved lower bound for $\mathrm{Keep}_i$ as follows. 
\begin{align}
\mathrm{Keep}_i  & \ge   \mathrm{exp}\left(  - \frac{ 1}{ (1- \frac{\delta}{k^{100k}})} \sum_{r=1}^{k-1}   \alpha^{r}   R_{i,r} \right)  \nonumber  \\
			   & \ge   \mathrm{exp}\left( -	 \frac{ 1}{ \left( 1+ \delta - \frac{ 1.2 \delta}{ k^{99}  } \right)^{ k-1 }     } \sum_{r=1}^{k-1} (1- \alpha C)^{r (i-1) }  \frac{ K^r \prod_{p=r}^{k-2 } (p+1)   } {  (\ln \Delta)^{r-1 }   C^{k-1-r } }     \right )  \label{B_bound},
\end{align}
for sufficiently large $\Delta$.

Recall that $\delta= (1+\epsilon)(k-1) -1$. Using~\eqref{B_bound} we get
\begin{eqnarray}
\prod_{j = 1}^{ i-1} \mathrm{Keep}_j 	   & \ge &  \mathrm{exp}\left( -	\frac{ 1}{ \left( 1+ \delta - \frac{ 1.2 \delta}{ k^{99}  } \right)^{ k-1 }     } \sum_{r=1}^{k-1}  \left( \frac{ K^r \prod_{p=r}^{k-2 } (p+1)   } {  (\ln \Delta)^{r-1 }   C^{k-1-r } }      \sum_{j=1}^{i-1}  (1 - \alpha C )^{r (j-1)}     \right ) \right)  \nonumber  \\
						    & \ge &  \mathrm{exp}\left( -	\frac{ 1}{ \left( 1+ \delta - \frac{ 1.2 \delta}{ k^{99}  } \right)^{ k-1 }     } \sum_{r=1}^{k-1}  \left( \frac{ K^r \prod_{p=r}^{k-2 } (p+1)   } {  (\ln \Delta)^{r-1 }   C^{k-1-r } }  \cdot  \frac{1}{ 1 -  (1- \alpha C)^r }    \right ) \right)  \nonumber  \\
						    & \ge &   \mathrm{exp}\left( -	\frac{ (k-1)! C^{-(k-1)} \ln \Delta }{ \left( 1+ \delta - \frac{ 1.2 \delta}{ k^{99}  } \right)^{ k-1 }   } \sum_{r=1}^{k-1}  \left( \left( \frac{ C\cdot K  } {  \ln \Delta  }  \right )^r  \cdot \frac{1}{ (1 - (1-\alpha C)^r)  } \right) \right)	 \nonumber \\
						    & \ge &   \mathrm{exp}\left( -	\frac{ (k-1)! C^{-(k-1)} \ln \Delta }{ \left( 1+ \delta - \frac{ 1.2 \delta}{ k^{99}  } \right)^{ k-1 }   } \sum_{r=1}^{k-1}  \left( \left( \frac{ C\cdot K  } {  \ln \Delta  }  \right )^r  \cdot \frac{1}{ \alpha C } \right) \right)	 \nonumber \\				     
						     & = &   \mathrm{exp}\left( -	\frac{ (k-1)! C^{-(k-1)} \ln \Delta }{ \left( 1+ \delta - \frac{ 1.2 \delta}{ k^{99}  } \right)^{ k-1 }    }  \cdot \frac{\frac{C \cdot K}{ \ln \Delta} - \left(\frac{C \cdot K}{ \ln \Delta}\right)^{k} }{1 - \frac{C \cdot K}{ \ln \Delta}   }  \cdot \frac{ \ln \Delta}{ C \cdot K }  \right)	 \nonumber \\
						     & \ge &   \mathrm{exp}\left( -	\frac{ (k-1)! C^{-(k-1)} \ln \Delta }{ \left( 1+ \delta - \frac{ 1.2 \delta}{ k^{99}  } \right)^{ k-1 }     }     \cdot  \left(1+ \frac{\delta}{k^{100}} \right) \right)	 \nonumber \\
						     & = & \mathrm{exp}\left( -	\frac{ (k-1)! C^{-(k-1)} \left(1+ \frac{\delta}{k^{100}} \right) \ln \Delta }{   \left( (1+\epsilon)(k-1) -  \frac{ 1.2 \delta}{ k^{99}  } \right)^{ k-1 }     }  \right)	 \nonumber \\
						     & = & \mathrm{exp}\left( -	\frac{ (k-1)! C^{-(k-1)} \left(1+ \frac{\delta}{k^{100}} \right) \ln \Delta }{ (1+\epsilon)^{k-1} \cdot (k-1)^{k-1}  \left( 1 -  \frac{ 1.2 \delta}{ k^{99} (1+\epsilon)(k-1)  } \right)^{ k-1 }     }  \right)	 \nonumber \\
						   & \ge & \mathrm{exp} \left(  - \frac{\ln \Delta}{ (1+\frac{\epsilon}{2} ) (k-1) }    \right)  	 \label{finally!!} 
\end{eqnarray}
for sufficiently large $\Delta$, and  since $\frac{(k-1)!  \gamma  }{ (1+\epsilon)^{k-1} (k-1)^{k-1} }  \cdot (1+ \frac{\epsilon}{2} ) (k-1)  < 1 $, where  $\gamma = \frac{(1+ \frac{\delta}{ k^{100} } ) C^{-(k-1)}  }{  \left( 1 -  \frac{ 1.2 \delta}{ k^{99} (1+\epsilon)(k-1)  } \right)^{ k-1 } } $. To see the latter inequality, note that $(k-1)! < (k-1)^k \mathrm{e}^{-k }$ and, therefore, for $k \ge 3$,
\begin{align*}
\frac{(k-1)!  \gamma  }{ (1+\epsilon)^{k-1} (k-1)^{k-1} }  \cdot \left(1+ \frac{\epsilon}{2} \right) (k-1)    <  \frac{1 + \frac{\epsilon }{ 2}  }{( 1+ \epsilon)^{k-1} }  \cdot e^{-k} (k-1) \gamma < \frac{1 + \frac{\epsilon }{ 2}  }{1+ \epsilon }  \cdot e^{-k}  (k-1)\gamma  <1.
\end{align*}

Using~\eqref{finally!!} we can now bound $L_i'$ as follows.
\begin{eqnarray}
L_i'   =   L_1' \prod_{j=1}^{i-1} \mathrm{Keep}_j \ge (1+\delta) \left( \frac{ \Delta }{ \ln \Delta } \right)^{\frac{1}{k-1} }   \Delta^{- \frac{ 1}{ (1+\frac{\epsilon}{2} ) (k-1) }  }   \ge   \Delta^{ \eta}, 
\end{eqnarray}
for sufficiently large $\Delta$. Thus, $L_i'$ never gets too small for the purposes of our analysis. Lemma~\ref{no_errors} implies that neither does $L_i$.

The proof is concluded by observing that~\eqref{Gs}  implies that $R_{i,r}$, $ r \in [k-1]$, becomes smaller than $\frac{1 }{10 k^2 }$ for $i = O( \ln \Delta \ln \ln \Delta)$.

\subsection{Proof of Lemma~\ref{bounding_keep_lemma}} 
\label{bounding_keep_lemma_proof}

We proceed by induction. The case $i=1$ is straightforward to verify  since $R_{1,r} = 0 $ for every $r \in [k-2]$, while  $R_{1,k-1} = \frac{ \ln \Delta }{ (1+\delta)^{k-1} } $. Therefore, we inductively assume the claim for $i$, and consider the case $i+1$. 
Note that the inductive hypothesis implies that $\mathrm{Keep}_i  = \Omega(1)$ since $ 1 - \frac{1}{x}  \ge \mathrm{e}^{- \frac{1}{x-1} } $ for every $x \ge 2$ and, thus,
\begin{align}\label{constant_bound}
\mathrm{Keep}_i&  \ge \mathrm{exp}\left( -  \sum_{r=1}^{ k-1}  \frac{T_{i,r}  }{ ( \alpha^{-1} L_i  ) ^r - 1 }  \right)  = \mathrm{exp} \left( -  \sum_{r=1}^{k-1} \frac{R_{i,r} }{ \left( \frac{ \ln \Delta}{ K} \right)^{r} -\frac{1}{L_i^r} } \right) \ge   \mathrm{exp} \left( -  \sum_{r=1}^{k-1} \frac{k^{2 (k-1-r )} \ln \Delta }{ \left( \frac{ \ln \Delta}{ K} \right)^{r} -\frac{1}{L_i^r} } \right)      \nonumber \\
 &\ge \mathrm{exp} \left( - \frac{ K k^{2(k-2)} \ln \Delta }{ \ln \Delta - \frac{K}{L_i}}  -O \left( \frac{1  }{ \ln \Delta } \right) \right) \ge   \mathrm{exp}\left( - \frac{ K k^{ 2(k-2)  } }{1 - \frac{\delta}{100k}  }      \right),
\end{align}
for sufficiently large $\Delta$.  Recalling~\eqref{L_def} and~\eqref{T_def}, we have:
\begin{eqnarray}
 R_{i+1,r} & = &  	 \sum_{j =  r }^{ k-1 } \left(  \frac{T_{i,j} }{L_{i+1}^r  }  \cdot \left(   \mathrm{Keep}_{i} \left(  1- \alpha \mathrm{Keep}_{i} \right)   \right)^{ r } { j \choose r}  \left(\frac{ \alpha \mathrm{Keep}_{i}}{L_{i} }\right)^{ j-r}  \right) \nonumber  \\
 		& &	+ \frac{1}{L_{i+1}^r} \left(\sum_{j=r}^{k-1}  { j \choose r}  \alpha^{j-r} \frac{T_{i,j} }{ L_i^{j-r} }\right)^{2/3}   +  4 k^{ 2(k-r) } \alpha^{-r+1}   \left( \frac{ L_i}{L_{i+1} } \right)^{r} \ln \Delta        \sum_{\ell=1}^{k-1}\frac{T_{i,\ell} }{ L_i^{2\ell} (\ln \Delta)^{2\ell} }    \nonumber \\
 		& =&  \sum_{j =  r }^{ k-1 } \left(  \frac{T_{i,j} }{L_{i}^r  \left(  \mathrm{Keep}_i  -  L_i^{-1/3} \right)^r }  \cdot \left(   \mathrm{Keep}_{i} \left(  1- \alpha \mathrm{Keep}_{i} \right)   \right)^{ r } { j \choose r}  \left(\frac{ \alpha \mathrm{Keep}_{i}}{L_{i} }\right)^{ j-r}  \right)  \nonumber \\
		& &   +O\left( \frac{1}{ \ln \Delta }  \right)  \label{bounding_errors_work} \\
		& =&  \sum_{j =  r }^{ k-1 } \left(  R_{i,j}  \cdot  \frac{\left(     1- \alpha \mathrm{Keep}_{i}   \right)^{ r }}{ \left( 1-  \frac{L_i^{-1/3}}{ \mathrm{Keep}_i }  \right)^r }  { j \choose r}  \left( \alpha \mathrm{Keep}_{i}\right)^{ j-r}  \right) + O\left( \frac{1}{ \ln \Delta }  \right)  \nonumber \\
		& \le&  \sum_{j =  r }^{ k-1 } \left(  R_{i,j}  \cdot \left(     1- \frac{\alpha \mathrm{Keep}_{i}}{2}   \right)^{ r } { j \choose r}  \left( \alpha \mathrm{Keep}_{i}\right)^{ j-r}  \right) +  O\left( \frac{1}{ \ln \Delta }  \right)   \label{talep1} \\
		&\le &  \left(     1- \frac{\alpha \mathrm{Keep}_{i}}{2}   \right)^{ r }   \left(  R_{i,r}  +  \sum_{j =  r +1  }^{ k-1 }  { j \choose r}  R_{i,j}  \alpha^{j-r}   \right) + O\left( \frac{1}{ \ln \Delta }  \right)  	 \nonumber \\
		&\le &  \left(     1- \frac{\alpha \mathrm{Keep}_{i}}{2}   \right)   \left(  k^{2 (k-1-r) }  \ln \Delta +  \sum_{j =  r +1  }^{ k-1 }  { j \choose r}   \frac{k^{2(k-1-j)} K^{j-r}} { (\ln \Delta)^{ j-r -1 }  }     \right) +O\left( \frac{1}{ \ln \Delta }  \right)  \nonumber
\end{eqnarray}
\begin{eqnarray}
		&\le & \left(     1- \frac{\alpha \mathrm{Keep}_{i}}{2}   \right)  \left(  k^{2(k-1-r)} \ln \Delta       +  k^{2(k-1-(r+1))} (r+1) K \right) + O\left( \frac{1}{ \ln \Delta }  \right)  \nonumber  \\
& \le &  k^{2(k-1-r) }  \ln \Delta  - K  \left(   \frac{    \mathrm{Keep}_i      k^{ 2(k-1-r) }   }{ 2}   - k^{ 2(k-1-(r+1)) } (r+1) \right) +  O\left( \frac{1}{ \ln \Delta }  \right)   \label{twenty_six_correction} \\
& \le &  k^{2(k-1-r) }  \ln \Delta,\label{final_decreasing_trick}
\end{eqnarray}
for sufficiently large $\Delta$, concluding the proof. Note that in deriving~\eqref{bounding_errors_work} we used the inductive hypothesis and that $L_i \ge (\ln \Delta)^{20(k-1)}$  to obtain:
\begin{align*}
\frac{1}{L_{i+1}^r} \left(\sum_{j=r}^{k-1}  { j \choose r}  \alpha^{j-r} \frac{T_{i,j} }{ L_i^{j-r} }\right)^{2/3}  & = \left(\frac{1}{L_{i+1}^r} \right)^{1/3 }  \left(\sum_{j=r}^{k-1}  { j \choose r}  \alpha^{j-r}  R_{i,j} \frac{ L_i^r }{ L_{i+1}^r }\right)^{2/3}  = o\left( \frac{ 1}{ \ln \Delta } \right), \\
4 k^{ 2(k-r) } \alpha^{-r+1}   \left( \frac{ L_i}{L_{i+1} } \right)^{r} \ln \Delta        \sum_{\ell=1}^{k-1}\frac{T_{i,\ell} }{ L_i^{2\ell} (\ln \Delta)^{2\ell} }       & = 4 k^{ 2(k-r) } \ln \Delta \cdot \alpha^{-r+1}   \left(\frac{L_i}{ L_{i+1}}\right)^r \sum_{ \ell =1}^{k-1} \frac{R_{i,\ell} }{ L_i^{\ell} (\ln \Delta)^{2\ell}  } \\
																											& = o\left( \frac{ 1}{ \ln \Delta } \right).
\end{align*}
In deriving~\eqref{talep1} we used the facts that $\mathrm{Keep}_i = \Omega(1)$, $ L_i, T_{i,r} \ge (\ln \Delta)^{20(k-1)}$.  In deriving~\eqref{final_decreasing_trick}  we used  the fact that   $\frac{    \mathrm{Keep}_i      k^{ 2(k-1-r) }   }{ 2}   - k^{ 2(k-1-(r+1)) } (r+1) > 0$,  since $k \ge 3$ and, according to~\eqref{constant_bound}, $\mathrm{Keep}_i$ is sufficiently close to $1$. 

\subsection{Proof of Lemma~\ref{pain} }
\label{pain_lemma_proof}

We proceed by induction.  The base cases are easy to verify since  $R_{1,r}' = 0$  for every $ r \in [k-2]$ and $R_{1,k-1}' = \frac{  \ln \Delta}{ (1+\delta)^{k-1}} $.

We first focus on the case $r = k-1$. We assume that the claim is true for $i-1$ and consider $i$. Note that the inductive hypothesis, the facts that  $L_{j} \ge (\ln \Delta)^{20(k-1) }  $ for every $1 < j < i$  and  $\mathrm{Keep}_j \ge C$,  imply:
\begin{align}\label{bounding_S}
4 k^{ 2(k-r) } \alpha^{-r+1}   \left( \frac{ L_i}{L_{i+1} } \right)^{r} \ln \Delta        \sum_{\ell=1}^{k-1}\frac{T_{i,\ell} }{ L_i^{2\ell} (\ln \Delta)^{2\ell} }  \le  \frac{1}{ (\ln \Delta)^{10(k-1) }  },
\end{align}
for sufficiently large $\Delta$, for every $r \in [k-1]$ and $1 <j \le i$. Therefore, recalling~\eqref{Lprime_def},~\eqref{Tprime_def} and using~\eqref{bounding_S}, we have:
\begin{eqnarray*}
R_{i,k-1}'  & \le&R_{i-1,k-1}'(1- \alpha \mathrm{Keep}_{i-1})^{k-1}  +  \frac{1}{ (\ln \Delta)^{10(k-1)} }  \\
		& \le& R_{i-2,k-1}'(1-\alpha \mathrm{Keep}_{i-2})^{k-1}  (1-\alpha \mathrm{Keep}_{i-1} )^{k-1} +  \frac{ (1-\alpha\mathrm{Keep}_{i-1})^{k-1}}{ (\ln \Delta)^{10(k-1)}  }  + \frac{1}{ (\ln \Delta)^{10(k-1)} }  \\
		&\le&R_{i-2,k-1}'(1-\alpha C)^{2(k-1)} +  \frac{ (1-\alpha C)^{k-1}}{ (\ln \Delta)^{10(k-1)}  }  + \frac{1}{ (\ln \Delta)^{10(k-1)} }  \\
		& \le & \ldots \nonumber \\
		& \le& (1-\alpha C)^{(i-1)(k-1)} R_{1,k-1}' + \frac{ 1}{  (\ln \Delta)^{10(k-1) } }  \sum_{\ell =0 }^{i-1} (1- \alpha C)^{(k-1) \ell} \nonumber  \\
		 & \le &(1-\alpha C)^{(i-1)(k-1)}  \frac{\ln \Delta}{ (1+\delta)^{k-1} }  +  \frac{1}{ (\ln \Delta)^{ 5(k-1) }  }  \nonumber \\
		& \le&   (1-\alpha C)^{(i-1)(k-1)}  \frac{   \ln \Delta  }{ (1+\delta - \frac{\delta }{k^{99} } )^{k-1} },
\end{eqnarray*}
for sufficiently large $\Delta$, concluding the proof for the case $r = k-1$.

We now focus on $r \in [k-2]$. We first observe that
\begin{eqnarray}
R_{2,r}'     & \le &   \sum_{j =  r }^{ k-1 } \left(  \frac{T_{1,j}'}{(L_{1}')^r  }  \cdot \left(   \mathrm{Keep}_{1} \left(  1- \alpha \mathrm{Keep}_{1} \right)   \right)^{ r } { j \choose r}  \left(\frac{ \alpha \mathrm{Keep}_{1}}{L_{1}' }\right)^{ j-r}  \right) +\frac{1}{ (\ln \Delta)^{10(k-1)} }   \nonumber    \\
		& =&  R_{1,k-1}'   \cdot \left(   \mathrm{Keep}_{1} \left(  1- \alpha \mathrm{Keep}_{1} \right)   \right)^{ r } { k-1 \choose r}  \left( \alpha \mathrm{Keep}_{1}\right)^{ k-1-r} +\frac{1}{ (\ln \Delta)^{10(k-1)} }    \nonumber    \\
		& \le  & \frac{ (\ln \Delta)^{ -(k-2-r) }  }{ (1+\delta)^{k-1}  }  K^{k-1-r} {k-1  \choose  r} +\frac{1}{ (\ln \Delta)^{10(k-1)} }  \nonumber  \\
		& \le  & \frac{ ( 1 + \frac{\delta }{k^{100} }  )^{k-1-r }     (\ln \Delta)^{ -(k-2-r) }  }{ (1+ \delta - \frac{\delta}{k^{99} } )^{k-1}  }  K^{k-1-r} \prod_{p = r}^{ k-2}  (p+1) \nonumber,
\end{eqnarray}
concluding the proof of the base cases.

Assume that the claim holds for all  pairs $(r',i')$, where $r'  \in \{r, \ldots, k-1\}$ and  $i' \le i-1 $. It suffices to prove that it also holds for any pair $(r,i)$, where $i > 2$ and $r \in [k-2]$. To see this, observe that 
\begin{eqnarray}
 R_{i,r}' & \le &  	 \sum_{j =  r }^{ k-1 } \left(  \frac{T_{i-1,j}'}{(L_{i}')^r  }  \cdot \left(   \mathrm{Keep}_{i-1} \left(  1- \alpha \mathrm{Keep}_{i-1} \right)   \right)^{ r } { j \choose r}  \left(\frac{ \alpha \mathrm{Keep}_{i-1}}{L_{i-1}' }\right)^{ j-r}  \right) +\frac{1}{ (\ln \Delta)^{10(k-1)} }    \nonumber \\
	  & = &  	 \sum_{j =  r }^{ k-1 } \left(  \frac{T_{i-1,j}'}{(L_{i-1}')^j }  \cdot   \mathrm{Keep}_{i-1}^{j-r}    \left(  1- \alpha \mathrm{Keep}_{i-1} \right)^{ r } { j \choose r}  \alpha^{j-r}  \right) +\frac{1}{ (\ln \Delta)^{10(k-1)} }    \nonumber  \\
	   & = &  \left(  1- \alpha \mathrm{Keep}_{i-1}    \right)^{ r }   \sum_{j =  r }^{ k-1 } \left(   R_{i-1,j}'   { j \choose r}  (\alpha \mathrm{Keep}_{i-1} )^{j-r}   \right) +\frac{1}{ (\ln \Delta)^{10(k-1)} }    \nonumber  \\
	   & \le &  \left(  1- \alpha C    \right)^{ r }   \sum_{j =  r }^{ k-1 } \left(   R_{i-1,j}'   { j \choose r}  \alpha^{j-r}   \right ) +\frac{1}{ (\ln \Delta)^{10(k-1)} }    \nonumber   \\
	        & \le &  \left(  1- \alpha C  \right)^{ r } R_{i-1,r}'    \nonumber \\
	        &  &+ \sum_{j =r+1}^{k-1 } { j \choose r}  K^{j-r}  (\ln  \Delta)^{1- (j-r)  }  (1-\alpha C)^{ j(i-2)+r}    \frac{ (1 + \frac{ \delta}{ k^{100}} )^{k-1-j }  }{ (1+\delta- \frac{\delta }{k^{99} })^{k-1 }  C^{k-1-j } }  \prod_{p=j}^{ k-2}(p+1)     \nonumber\\
	    &   & \hphantom{\left(  1- \alpha C \right)^{ 2r }asdfasfdsadfasdfasdfasdfsadfasdfadfasdfasfdfasdf}		+\frac{1}{ (\ln \Delta)^{10(k-1)} }    \nonumber   \\
	      & \le &\left(  1- \alpha C \right)^{ 2r }  R_{i-2,r}'   \nonumber \\
	     && + \sum_{j =r+1}^{k-1 } { j \choose r}  K^{j-r}  (\ln  \Delta)^{1- (j-r)  }  (1-\alpha C)^{ j(i-2)+r}    \frac{ (1 + \frac{ \delta}{ k^{100}} )^{k-1-j }  }{ (1+\delta- \frac{\delta }{k^{99} })^{k-1 }  C^{k-1-j } }  \prod_{p=j}^{ k-2}(p+1)   \nonumber	\\
	     &  &+  \sum_{j =r+1}^{k-1 } { j \choose r}  K^{j-r}  (\ln  \Delta)^{1- (j-r)  }  (1-\alpha C)^{ j(i-3)+2r}    \frac{ (1 + \frac{ \delta}{ k^{100}} )^{k-1-j }  }{ (1+\delta- \frac{\delta }{k^{99} })^{k-1 }   C^{k-1-j }}  \prod_{p=j}^{ k-2}(p+1)     \nonumber \\
	     &  &  +\frac{1 + (1-\alpha C)^r}{ (\ln \Delta)^{10(k-1)} } \label{again} 
\end{eqnarray}
\begin{eqnarray}    	    
	       \hphantom{ R_{i,r}' } 	      & \le& \left(  1- \alpha C \right)^{ 3r }  R_{i-3,r}' \nonumber  \\
             && + \sum_{j =r+1}^{k-1 } { j \choose r}  K^{j-r}  (\ln  \Delta)^{1- (j-r)  }  (1-\alpha C)^{ j(i-2)+r}    \frac{ (1 + \frac{ \delta}{ k^{100}} )^{k-1-j }  }{ (1+\delta- \frac{\delta }{k^{99} })^{k-1 }  C^{k-1-j } }  \prod_{p=j}^{ k-2}(p+1)   \nonumber	\\
	     &  &+  \sum_{j =r+1}^{k-1 } { j \choose r}  K^{j-r}  (\ln  \Delta)^{1- (j-r)  }  (1-\alpha C)^{ j(i-3)+2r}    \frac{ (1 + \frac{ \delta}{ k^{100}} )^{k-1-j }  }{ (1+\delta- \frac{\delta }{k^{99} })^{k-1 }  C^{k-1-j } }  \prod_{p=j}^{ k-2}(p+1)     \nonumber\\
	      &  &+  \sum_{j =r+1}^{k-1 } { j \choose r}  K^{j-r}  (\ln  \Delta)^{1- (j-r)  }  (1-\alpha C)^{ j(i-4)+3r}    \frac{ (1 + \frac{ \delta}{ k^{100}} )^{k-1-j }  }{ (1+\delta- \frac{\delta }{k^{99} })^{k-1 }  C^{k-1-j } }  \prod_{p=j}^{ k-2}(p+1)     \nonumber\\
	      &  & +\frac{1 + (1-\alpha C)^{r} +  (1-\alpha C)^{2r} }{ (\ln \Delta)^{10(k-1)} }	\label{again_again} \\
	    	      &\le & \ldots 	\nonumber \\
	      & \le&\left(  1- \alpha C \right)^{ (i-1)r }  R_{1,r}' \nonumber \\
	      & & +   \sum_{ j = r+1 }^{k-1 } { j \choose r }  K^{ j-r}  (\ln \Delta)^{1 - (j-r) } \prod_{p = j }^{k-2 } (p+1)  \frac{( 1 + \frac{\delta}{ k^{100} })^{k-1-j }  }{ (1 +\delta- \frac{\delta }{k^{99} })^{ k-1}  C^{k-1-j } }   \sum_{ \ell = 1 }^{i-1  } (1-\alpha C)^{j(i-\ell -1) + \ell r }       \nonumber \\
	      &  & +\frac{ \sum_{\ell =0}^{i-2 } (1 - \alpha C)^{r \ell}  }{ (\ln \Delta)^{10(k-1)} } \  \label{again_again_again} \\
	       & \le&    \sum_{ j = r+1 }^{k-1 } { j \choose r }  K^{ j-r}  (\ln \Delta)^{1 - (j-r) } \prod_{p = j }^{k-2 } (p+1)  \frac{( 1 + \frac{\delta}{ k^{100} })^{k-1-j }  }{ (1 +\delta- \frac{\delta }{k^{99} })^{ k-1}  C^{k-1-j}   }   \sum_{ \ell = 1 }^{i-1  } (1-\alpha C)^{j(i-\ell -1) + \ell r }    \nonumber \\
	      & &  + O \left( \frac{1}{ (\ln \Delta)^{5(k-1) } }  \right)   \label{tricksss} \\
	      & = & \sum_{ j = r+1 }^{k-1 } { j \choose r }  K^{ j-r}  (\ln \Delta)^{1 - (j-r) } \prod_{p = j }^{k-2 } (p+1)  \frac{( 1 + \frac{\delta}{ k^{100} })^{k-1-j }  }{ (1 +\delta- \frac{\delta }{k^{99} })^{ k-1}  C^{k-1-j } }   \sum_{ \ell = 1 }^{i-1  } (1-\alpha C)^{ (i-1)r + (i-\ell-1)(j-r)  }   \nonumber\\
	       & &  + O \left( \frac{1}{ (\ln \Delta)^{5(k-1) } }  \right)   \label{pre_error}\\
	       & \le &         \frac{\left(  1- \alpha C \right)^{ (i-1)r } }{ (1+\delta- \frac{\delta }{k^{99} })^{k-1} }  		\sum_{j=r+1 }^{k-1 } { j \choose r } K^{  j-r}   (\ln \Delta)^{1 -( j-r) } \prod_{p =j}^{k-2 }(p+1)  \frac{(1+ \frac{\delta }{k^{100}} )^{k-1-j}}{ C^{k-1-j } }   \sum_{\ell \ge 0 } (1-\alpha C)^{  \ell (j-r)  }   \label{error_avoidance} 		\\	 
		   & = &         \frac{\left(  1- \alpha C \right)^{ (i-1)r } }{ (1+\delta- \frac{\delta }{k^{99} })^{k-1} } 	\left( 		\sum_{j=r+1 }^{k-1 } { j \choose r } K^{  j-r}   (\ln \Delta)^{1 -( j-r) } \frac{\prod_{p =j}^{k-2 }(p+1)}{ C^{k-1-j } }   \frac{(1+  \frac{\delta }{k^{100} })^{k-1-j}}{ 1 - (1 - \alpha C )^{j-r} }    	\right)	\label{finally_two}		\\	 
	       	&\le &		 (1 - \alpha C )^{r (i-1)} \ln \Delta   \cdot     \frac{   (1+  \frac{\delta}{k^{100}} )^{k-1-r}  }{ (1+\delta- \frac{\delta }{k^{99} })^{k-1} C^{k-1-r }   } \prod_{p=r}^{ k-2}(p+1),		\label{final_finally_yes}
\end{eqnarray}
for sufficiently large $\Delta$, concluding the proof.  Note that in order to get~\eqref{again} we  upper bound $R_{i-1,r}'$ in the same way we upper bounded $R_{i,r}'$. We keep using the same steps to bound $R_{i-2,r}', R_{i-3,r}', \ldots$ until we  get~\eqref{again_again_again}. In deriving~\eqref{tricksss} 
we used that $R_{1,r} = 0 $ for every $r \in [k-2]$. In going from~\eqref{pre_error} to~\eqref{error_avoidance} we start the summation in the last term    from $\ell = 0$ instead from $\ell =1$ in order to subsume the  $O ( 1/ (\ln \Delta)^{5(k-1) }  )$ error term.   Finally, in going from~\eqref{finally_two} to~\eqref{final_finally_yes} 
we multiply the term that corresponds to $j= r+1$ in the summation by $(1+  \frac{\delta}{k^{100}} )  $ in order to subsume the terms of the summation that correspond to $j > r+1$.

\section{A sufficient pseudo-random property for coloring}\label{properties}

In this section we present the proof of Theorem~\ref{main_property}. To do so, we build on ideas of   Alon, Krivelevich and Sudakov~\cite{alon1999list} and show that the random hypergraph $H( k,n, d/{ n \choose k-1 })$ asymptotically almost surely admits a few useful  features.

The first lemma  we prove states that all  subgraphs of $H(k,n,d/{n \choose k-1 } )$ with not too many vertices are sparse and, therefore, of small degeneracy.
\begin{lemma}\label{simple_properties}
For every constant $k \ge 2 $,  there exists $d_{k} > 0$ such that for any  constant $d \ge d_{k}$,  the random hypergraph $H(k,n,d/ {n \choose  k-1 } )$ has the following property asymptotically almost surely: Every $s \le  n  d^{-\frac{1}{k-1} } $  vertices of $H$ span fewer than $s \left(\frac{ d }{  (\ln d)^2 }\right)^{ \frac{1}{k-1} }   $ hyperedges. Therefore, any subhypergraph of $H$ induced by a subset $V_0 \subset V$ of size $|V_0 |  \le n  d^{-\frac{1}{k-1} }$,
is $k \left(\frac{ d  }{ (\ln d)^2 } \right)^{\frac{1}{k-1} } $-degenerate.
\end{lemma}
\begin{proof}
Given the statement about the sparsity of any subhypergraph of $H$ incudced by a set of $s \le n d^{-\frac{1}{k-1} } $ vertices, the claim about its degeneracy  follows from the fact that its average (and, therefore its minimum) degree is at most
\begin{align*}
k \cdot \frac{ s \left(\frac{ d }{  (\ln d)^2 }\right)^{ \frac{1}{k-1} }    }{ s    } =  k \cdot  \left(\frac{ d }{  (\ln d)^2 }\right)^{ \frac{1}{k-1} }.
\end{align*}
So we are left with proving the statement regarding sparsity.

Letting $r  =  \left( \frac{  d}{ (  \ln  d)^2 } \right)^{ \frac{1}{k-1} }  $, we see that the  probability that there exists a subset $V_0 \subset V$ which violates the  statement of the lemma is at most
\begin{eqnarray}
\sum_{i=r^{ \frac{1}{k-1} } }^{  n  d^{-\frac{1}{k-1} } }  {n \choose i } {  {i \choose k } \choose r i  } \left( \frac{ d}{ { n   \choose k-1} }  \right)^{r i }  & \le & \sum_{ i=r^{ \frac{1}{k-1} } }^{ n  d^{-\frac{1}{k-1} }} \left[ \frac{ \mathrm{e} n}{ i }  \left( \frac{  \mathrm{e} i^{k-1} }{ r} \right)^r        \left(\frac{d }{ {n \choose k-1 }} \right)^{r }  \right]^i    \label{assertion_one}  \\
&= & \sum_{ i=r^{ \frac{1}{k-1} } }^{ n  d^{-\frac{1}{k-1} }} \left[ \frac{ \mathrm{e} n}{ i } \cdot \left( \frac{ \mathrm{e} i^{k-1} d }{r {n \choose k-1 }  }  \right)^{\frac{1}{k-1 }  }  \cdot \left( \frac{ \mathrm{e} i^{k-1} d }{r {n \choose k-1 }  }  \right)^{r- \frac{1}{k-1 }  } \right]^i  \nonumber \\
&\le & \sum_{i=r^{ \frac{1}{k-1} } }^{ n  d^{-\frac{1}{k-1} }} \left[  \mathrm{e}^{1 + \frac{1}{k-1} } (k-1)   \left(\frac{d}{r} \right)^{ \frac{1}{k-1} }   \left( \frac{ \mathrm{e} i^{k-1} d }{r {n \choose k-1 }  }  \right)^{r- \frac{1}{k-1 }  }  \right]^i  \nonumber  \\  
&\le & \sum_{i=r^{ \frac{1}{k-1} } }^{ n  d^{-\frac{1}{k-1} }} \left[  \mathrm{e}^{1 + \frac{1}{k-1} } (k-1)   \left(\frac{d}{r} \right)^{ \frac{1}{k-1} }   \left( \frac{ \mathrm{e}  d (k-1)^{k-1}}{r    } \cdot  \left(\frac{i  }{n }\right)^{k-1}   \right)^{r- \frac{1}{k-1 }  }   \right]^i  \nonumber  \\  
& =& \sum_{i=r^{ \frac{1}{k-1} } }^{ n  d^{-\frac{1}{k-1} }} \left[  \mathrm{e}^{1 + \frac{1}{k-1} } (k-1)   \left(\frac{d}{r} \right)^{ \frac{1}{k-1} }   \left( \frac{ \mathrm{e}  d (k-1)^{k-1}}{r    } \cdot  \left(\frac{i  }{n }\right)^{k-1}   \right)^{r- \frac{1}{k-1 }  }   \right]^i  \nonumber  
\end{eqnarray}
\begin{eqnarray}
& = & \sum_{i=r^{ \frac{1}{k-1} } }^{ n  d^{-\frac{1}{k-1} }} \left[  \mathrm{e}^{r+1} (k-1)^{r(k-1) }  d^r   r^{-r} \left(\frac{i  }{n }  \right)^{r(k-1) -1  }   \right]^i \nonumber \\
& =&  \sum_{i=r^{ \frac{1}{k-1} } }^{ n  d^{-\frac{1}{k-1} }} \left[  \mathrm{e} \left( \mathrm{e}  (k-1)^{(k-1) }  d   r^{-1} \right)^r \left(\frac{i  }{n }  \right)^{r(k-1) -1  }   \right]^i \nonumber \\
&= & \sum_{i=r^{ \frac{1}{k-1} } }^{ n  d^{-\frac{1}{k-1} }} \left[ A(k,d) \left(\frac{i  }{n }  \right)^{r(k-1) -1  }   \right]^i, \label{magka_ti_trexei}
\end{eqnarray}
where $A(k,d) =   \mathrm{e} \left( \mathrm{e}  (k-1)^{(k-1) }  d   r^{-1} \right)^r $. Note that in the lefthand side of~\eqref{assertion_one}  we used the fact that any subset of vertices of size $s < r^{ \frac{1}{ k-1} } $ cannot violate the assertion of the lemma, since it can span at most ${s \choose k } \le  s^k = s^{k-1}\cdot s < r s$ hyperedges. Moreover,throughout our derivation we used  that  for any pair of positive integers $\alpha, \beta$,  we have  $\left( \frac{ \alpha}{ \beta }  \right)^{\beta }  \le {\alpha \choose \beta }  < \left(  \frac{\alpha \cdot \mathrm{e}  }{ \beta } \right)^{\beta }   $.

We claim now that the function:
\begin{align}
F(i, d, k, n) =  \left[  A(k,d) \left(\frac{i  }{n }  \right)^{r(k-1) -1  }   \right]^i
\end{align}
is decreasing in $i$ when $i \in [r^{\frac{1}{k-1}}, n d^{-\frac{1}{k-1}} ]$ . To prove this, we examine the partial derivative of 
\begin{align*}
\ln(F(i,d,k,n) ) =  i \left(  \ln A(k,d) +   \left(r(k-1) - 1 \right)  \ln  \frac{i}{n} \right)
\end{align*}
with respect to $i$. In particular, we have that:
\begin{align*}
\frac{\partial }{ \partial i } ( \ln F(i,r,k,n) )   &= \ln A(k, d)  +  (r(k-1)-1) \left( \ln \frac{  i}{   n } + 1 \right) \nonumber \\
							       & \le  \ln  A(k,d) +  (r(k-1)-1) \left( - \frac{1}{k-1}  \ln d + 1 \right) \nonumber \\
							       & =  1 + r  \left( 1+ (k-1) \ln (k-1) + \ln \frac{d}{r} \right) + r \left( - \ln d + (k-1) \right) +\left( \frac{1}{k-1} \ln d -1 \right) \nonumber \\
							       & =  \frac{1}{k-1} \ln d  + r \left( 1+ (k-1) \ln (k-1) + (k-1) - \ln r  \right) \nonumber \\
							       & =  \frac{1}{k-1} \ln d  - \left( \frac{d }{ (\ln d)^2 } \right)^{ \frac{1}{k-1} }\left( \frac{1}{k-1} \ln \left( \frac{ d}{ ( \ln d )^2  }  \right)       - 1 - (k-1) ( 1 + \ln (k-1) ) \right)
\end{align*}
which is negative for sufficiently large $d$. This concludes the proof of our statement regarding the monotonicity of $F(i,d,k,n)$.

The latter monotonicity statement implies that the dominating term in the sum of~\eqref{magka_ti_trexei} is the term corresponding to $i =r^{ \frac{1}{k-1} } $. Therefore,~\eqref{magka_ti_trexei} implies that
\begin{align*}
\sum_{i=r^{ \frac{1}{k-1} } }^{  n  d^{-\frac{1}{k-1} } }  {n \choose i } {  {i \choose k } \choose r i  } \left( \frac{ d}{ { n   \choose k-1} }  \right)^{r i } < n \cdot  \left(  A(k,d)  \left(  \frac{ r^{ \frac{1}{k-1}}     }{ n } \right)^{r(k-1) - 1 }  \right)^{r^{\frac{1}{k-1}} }  = o(1)
\end{align*}
for sufficiently large $d$, concluding the proof.

\end{proof}

Next we  show that, that for any constant $c$, the number of vertices of  $H( k,n, d/{ n \choose k-1 }  )$ that have  degree $c$ essentially behaves as a Poisson random variable with mean $d$.

\begin{lemma}\label{degree_dist} 
For constants $c \ge1$, $k \ge 2$ and $d$ sufficiently large,  let $X_c$ denote the number of vertices of degree $c$ in $H(k,n,  d/ {n \choose k -1 }   )$. Then, asymptotically almost surely,
\begin{align*}
X_c  \le \frac{ d^{ c} \mathrm{e}^{- d}  }{ c! } n \left( 1 + O\left(  \frac{ \log n}{ \sqrt{n} }  \right) \right). 
\end{align*}
\end{lemma}

\begin{proof}
The upper bound on the expectation follows from standard ideas for estimation of the degree distribution of random graphs (see for example the proof of Theorem 3.3 in~\cite{frieze2016introduction}  for the case $k=2$). Assume that the vertices of $H(k,n,d/\binom{n}{k-1})$ are labeled $1,2,...,n$. Then,

\begin{align}
\mathbb{E}[X_{c}] &= n \Pr[\mathrm{deg}(1)=c]  \nonumber \\
&= n \binom{\binom{n-1}{k-1}}{c} \left(\frac{d}{\binom{n}{k-1}}\right)^{c} \left(1-\frac{d}{\binom{n}{k-1}}\right)^{\binom{n-1}{k-1}-c} \nonumber  \\
&\le n \frac{\binom{n-1}{k-1}^{c}}{c!} \left(\frac{d}{\binom{n}{k-1}}\right)^{c} \exp\left(- \left(\binom{n-1}{k-1}-c\right) \frac{d}{\binom{n}{k-1}}\right) \label{first_ineq_rand}  \\
&\le n \frac{d^{c}}{c!} \exp\left(- d \left(1 - \frac{k-1}{n} \right) + \frac{d c }{ \binom{n}{k-1}}     \right) \label{second_ineq_rand}   \\
&= n \frac{d^{c}e^{-d}}{c!}  \exp\left(\frac{d(k-1)}{n}  + \frac{d c }{ \binom{n}{k-1}}\right) \nonumber \\
& \le n \frac{d^{c}e^{-d}}{c!} \left(1+O\left(\frac{1}{n}\right)\right) \nonumber.
\end{align}
Note that in deriving~\eqref{first_ineq_rand} we used the fact that for every pair of positive integers $\alpha, \beta$ we have ${\alpha \choose \beta } \le \frac{ \alpha^{\beta} }{ \beta! } $. In deriving~\eqref{second_ineq_rand} we used the fact that
\begin{align*}
\frac{\binom{n-1}{k-1} }{ \binom{n}{k-1} }  = \frac{  \frac{ (n-1)! }{( n-1-(k-1))! (k-1)!   }   }{ \frac{ n! }{( n-(k-1))! (k-1)!   }} = \frac{n-(k-1) }{n} = 1 - \frac{k-1}{n}.
\end{align*}

To show concentration of $X_{c}$ around its expectation, we will use Chebyshev's inequality. In order to do so, we need to estimate the second moment of $X_{c}$ and, in particular, the joint probability $\Pr[ \mathrm{deg}(1)=\mathrm{deg}(2)=c]$. Let $N = \binom{n-1}{k-1}$ be the number of potential hyperedges containing a specific vertex, and $M = \binom{n-2}{k-2}$ be the number of potential hyperedges containing two specific vertices. Let $p=\frac{d}{\binom{n}{k-1}}$.

The joint probability can be written by conditioning on the number of hyperedges $l$ that contain both vertex 1 and vertex 2:

\begin{equation}
\Pr[\mathrm{deg}(1)=c, \mathrm{deg}(2)=c] = \sum_{l=0}^{c} \binom{M}{l} p^l (1-p)^{M-l} \left[ \binom{N-M}{c-l} p^{c-l} (1-p)^{(N-M)-(c-l)} \right]^2.
\end{equation}

We split this summation into the term for $l=0$ and the terms for $l \ge 1$.

\textbf{Case $l=0$:}
In this case, vertices 1 and 2 share no hyperedges. The corresponding term is:
$$
T_0 = 1 \cdot (1-p)^{M} \left[ \binom{N-M}{c} p^{c} (1-p)^{N-M-c} \right]^2.
$$
Notice that $\Pr[\mathrm{deg}(1)=c] = \binom{N}{c} p^c (1-p)^{N-c}$. Since $M = \Theta(n^{k-2})$ and $N = \Theta(n^{k-1})$, the ratio $\frac{\binom{N-M}{c}}{\binom{N}{c}}$ is $1+O(1/n)$. Similarly, $(1-p)^{-M} = 1+O(1/n)$. Therefore, we can express $T_0$ as:
$$
T_0 = (\Pr[\mathrm{deg}(1)=c])^2 \left(1 + O\left(\frac{1}{n}\right)\right).
$$
Crucially, the constant factor here is exactly 1, with an error term of order $O(1/n)$.

\textbf{Case $l \ge 1$:}
For $l \ge 1$, we observe that $\binom{M}{l} \le M^l$ and $\binom{N-M}{c-l} \le N^{c-l}$. The term $T_l$ in the summation is bounded by:
\begin{align*}
T_l & \le M^l p^l \cdot \left(  {N-M \choose c }  p^{c-l} \right]^2 \cdot (1-p)^{2(N-c) -M  + \ell} \\
    & = M^l p^l (1-p)^{-M+\ell }  \cdot \left( \left(1+O\left(\frac{1}{n}\right) \right) { N \choose c}  p^c (1-p)^{N-c}  \right)^2										\\
      &  = O(n^{(k-2)l}) \cdot O(n^{-(k-1)l}) \left(1 + O\left( \frac{1}{ n} \right) \right) \cdot (\Pr[\mathrm{deg}(1)=c])^2 \\
       & = O(n^{-l}) (\Pr[\mathrm{deg}(1)=c])^2.
\end{align*}
Summing over $l \ge 1$, the contribution is dominated by $l=1$, which is $O(1/n)$ relative to the squared probability. Combining these cases, we obtain:
\begin{align}
\Pr[\mathrm{deg}(1)=c, \mathrm{deg}(2)=c] &= T_0 + \sum_{l=1}^c T_l \nonumber \\
&= (\Pr[\mathrm{deg}(1)=c])^2 \left(1 + O\left(\frac{1}{n}\right)\right) + O\left(\frac{1}{n}\right)(\Pr[\mathrm{deg}(1)=c])^2 \nonumber \\
&= \Pr[\mathrm{deg}(1)=c] \Pr[\mathrm{deg}(2)=c] \left(1 + O\left(\frac{1}{n}\right)\right). \label{joint_pair_approx}
\end{align}

Now, letting $I_{j}$ denote the indicator random variable which equals 1 if vertex $j$ has degree $c$ and 0 otherwise, we calculate the variance:
\begin{align}
\mathrm{Var}[X_{c}] &= \mathbb{E}[X_{c}^{2}] - (\mathbb{E}[X_{c}])^{2}  \nonumber \\
& =  \mathbb{E}\left[ \left( \sum_{j=1}^{n} I_j \right)^2\right]  - (\mathbb{E}[X_{c}])^{2}   \nonumber \\
& = \mathbb{E} \left[ \sum_{i=1}^n \sum_{j=1}^n I_i I_j \right]  - \left(\mathbb{E}\left[ \sum_{i=1}^n I_i  \right] \right)^{2}  \nonumber  \\
&= \sum_{i=1}^{n} \sum_{j=1}^{n} (\Pr[\mathrm{deg}(i)=c, \mathrm{deg}(j)=c] - \Pr[\mathrm{deg}(i)=c]\Pr[\mathrm{deg}(j)=c]) \nonumber \\
&\le \sum_{i=1}^n \sum_{j=1}^n \Pr[\mathrm{deg}(i)=c]\Pr[\mathrm{deg}(j)=c] \left(1 + O\left(\frac{1}{n}\right)\right) - \Pr[\mathrm{deg}(i)=c]\Pr[\mathrm{deg}(j)=c]    \nonumber\\  
& + \sum_{i=1}^n \Pr[\mathrm{deg}(i)=c ] - (\Pr[ \mathrm{deg}(i) = c ])^2   \label{jello_above}  \\
& \le O\left(  \frac{1}{n} \right)\sum_{i=1}^n \sum_{j=1}^n \Pr[\mathrm{deg}(i)=c]\Pr[\mathrm{deg}(j)=c]   + \sum_{i=1}^n \Pr[ \mathrm{deg}(i) = c ]  \nonumber    \\
&\ =  O\left(  \frac{1}{n} \right)\left(\mathbb{E}[X_c]\right)^2  +   \mathbb{E}[X_c] \nonumber   \\
&\le A n \nonumber
\end{align}
for some constant $A=A(c,d)$, since $\mathbb{E}[X_c] = O(n)$. Note than in deriving~\eqref{jello_above} we used~\eqref{joint_pair_approx}.

Finally, applying Chebyshev's inequality, we obtain that for any $t>0$,
$$
\Pr[|X_{c}-\mathbb{E}[X_{c}]| \ge t\sqrt{n}] \le \frac{\mathrm{Var}[X_c]}{t^2 n} \le \frac{A}{t^{2}}.
$$
The proof is concluded by choosing $t=\log n$, which implies $X_c \le \frac{d^c e^{-d}}{c!} n (1+o(1))$ asymptotically almost surely.
\end{proof}

 Lemma~\ref{degree_dist} implies the following useful corollary.
\begin{corollary}\label{degree_dist_corollary}
For any constants $\delta  \in (0,1), k \ge 2, d > 0$,    let   $X = X(\delta,  k, d) $ denote the random variable equal to the number of vertices in $H(k,n, \ d/ {n \choose k -1 }   )$ whose degree is in $[ (1+\delta) d, 3 (k-1)^{k-1} d ]$. There exists a constant $d_{\delta}  > 0$ such that  if $d \ge d_{\delta}$ then, asymptotically almost surely,  $X  \le  \frac{ n }{ d^2}.$
\end{corollary}
\begin{proof} Let $X_{r}$ denote the number of vertices of degree $r$ in $H(k,n,d/\binom{n}{k-1})$. Since $k,d$ are constants, using Lemma 5.2 and Stirling's approximation we see that, asymptotically almost surely: \begin{align*} \sum_{r=(1+\delta)d}^{3(k-1)^{k-1}d}X_{r} &\le n\left(1+O\left(\frac{\log n}{\sqrt{n}}\right)\right) \sum_{r=(1+\delta)d}^{3(k-1)^{k-1}d} \frac{d^{r}e^{-d}}{r!} \\ &\le n(1+o(1)) \sum_{r=(1+\delta)d}^{3(k-1)^{k-1}d} \frac{d^{r}e^{-d}}{\sqrt{2\pi r}(r/e)^{r}} \\ &\le n \sum_{r=(1+\delta)d}^{\infty} \left( \frac{de}{r} \right)^r e^{-d} \end{align*} We observe that the terms in the sum are maximized at the lower bound $r = (1+\delta)d$. Substituting this value into the general term yields a geometric bound: \begin{align*} \left( \frac{de}{(1+\delta)d} \right)^{(1+\delta)d} e^{-d} &= \left( \frac{e}{1+\delta} \right)^{(1+\delta)d} e^{-d} \\ &= \left( \frac{e^{1+\delta}}{(1+\delta)^{1+\delta}} \right)^d e^{-d} \\ &= \left( \frac{e^{\delta}}{(1+\delta)^{1+\delta}} \right)^d.\end{align*} Using the standard inequality $(1+\delta)^{1+\delta} > e^{\delta}$ for $\delta > 0$, we set $\gamma = \frac{e^{\delta}}{(1+\delta)^{1+\delta}} < 1$. The sum is thus dominated by a geometric series with ratio less than 1, bounded by: $$ n \cdot B \cdot \gamma^d $$ for some constant $B$. Since $\gamma < 1$, for sufficiently large $d$ we have $B \gamma^d \le \frac{1}{d^2}$, and thus: $$ \sum_{r=(1+\delta)d}^{3(k-1)^{k-1}d}X_{r} \le \frac{n}{d^2}. $$ \end{proof}

Using Lemma~\ref{simple_properties} and Corollary~\ref{degree_dist_corollary} we show that, asymptotically almost surely,  only a small fraction of vertices of $H(k ,n, d/{n \choose k-1}  )$ have degree that significantly exceeds its average degree.
\begin{lemma}\label{part_b_random}
For every constants $k \ge 2 $ and $\delta  \in (0,1) $,  there exists $d_{k,\delta} > 0$ such that for any  constant $d \ge d_{k,\delta}$,  all but at most $\frac{2n } {  d^2 }$ vertices of   the random hypergraph $H(k,n,d/ {n \choose  k-1 } )$  have degree at most $(1+\delta) d$, asymptotically almost surely.
\end{lemma}
\begin{proof}
Corollary~\ref{degree_dist_corollary} implies that the number of vertices with degree in the interval $[(1+\delta)d, 3 (k-1)^{k-1} d] $ is at most $\frac{n}{d^2}$, for sufficiently large $d$. 

Suppose now  there are more than $ \frac{n}{ d^2} $ vertices with degree at least $3 (k-1)^{k-1} d$. Denote by $S$ a set containing exactly $\frac{n}{d^2} $ such vertices. According to Lemma~\ref{simple_properties}, asymptotically almost surely, the induced subhypergraph $H[S]$ 
has at most 
\begin{align*}
e(H[S] ) \le  \left(\frac{d }{ (\ln d)^2 } \right)^{\frac{1}{k-1} }  |S| =   \frac{ n}{ d^{2 - \frac{1}{ k-1}  }   (\ln d)^{ \frac{2}{k-1 }  }  }    
\end{align*}
 hyperedges. Therefore, the number of hyperedges between the sets of vertices $S$ and $V \setminus S$ is at least 
 \begin{align*}
3 (k-1)^{ k-1} d |S| - k e (H[S]  )   \ge  \frac{2.9   (k-1)^{k-1}  n}{d }  =: N.
 \end{align*}
 for sufficiently large $d$. However, the probability that $H( k,n, d/{n \choose k-1 } )$ contains such a  subhypergraph is at most
\begin{align*}
{n \choose \frac{ n}{ d^2 }  }  { \frac{ n^k}{ d^2 }   \choose  N} \left( \frac{d}{{n \choose k-1 } }  \right)^{ N}  \le   \left( \mathrm{e} d^2 \right)^{\frac{n }{ d^2  }} 
\left( \frac{ n^k \mathrm{e}  }{ d^2 N}  \cdot \frac{d }{  {n \choose k-1 }    }  \right)^{ N}  = o(1),
\end{align*}
for sufficiently large $d$. Note that   in deriving the final equality we used that for any pair of integers $\alpha, \beta$,  we have that ${\alpha \choose \beta }  \ge \left( \frac{ \alpha}{ \beta }  \right)^{\beta } $.  Therefore, asymptotically almost surely there are at most $\frac{n }{ d^2 } $ vertices in $G$ with degree greater than $3 (k-1)^{k-1} d$, concluding the proof.
\end{proof}

Finally, we show that the neighborhood of a typical vertex of $H(k,n, d/{n \choose k-1 }  )$ is locally  tree-like.
\begin{lemma}\label{part_c_random}
For every constants $k\ge 2, \delta  \in (0,1) $, asymptotically almost surely,  the random hypergraph $H(k,n,d/ {n \choose  k-1 } )$  has a subset $U \subseteq V(H)$ of size at most $n^{1-\delta}$ such that the induced hypergraph $H[V \setminus U] $ is of  girth at least $5$. 
\end{lemma}
\begin{proof}

Let $Y_2, Y_3, Y_4$, denote  the  number of $2$-, $3$- and $4$-cycles in $H(n,k,  d/ {n \choose k-1 }  )$, respectively. A straightforward calculation reveals that for $i \in \{2,3,4\}$:
\begin{align*}
\ex[Y_i ]   \le   \sum_{ s = 1 }^{ i(k-1)}   {n \choose s}  { { s \choose k-1 }   \choose   i }   \left( \frac{ d}{ {n \choose k-1 }   } \right)^{ i} &  \le  i (k-1)  \left( \frac{ n  \cdot \mathrm{e}}{ i (k-1) }  \right)^{ i (k-1) } \cdot  \left( \frac{{i(k-1) \choose k-1 }  }{ i} \right)^{i}  \cdot  \left( \frac{ d}{ {n \choose k-1 }   } \right)^{ i} \\
															& \le  i (k-1)  \left( \frac{ n  \cdot \mathrm{e}}{ i (k-1) }  \right)^{ i (k-1) }  \cdot \left( 	\frac{(i \cdot \mathrm{e} )^{k-1} }{ i} \right)^i \cdot \left(   \frac{d   (k-1)^{k-1}}{ n^{k-1} }   \right)^{i} \\
															& = i (k-1) \left(   \frac{ \mathrm{e}^{2(k-1) } d  }{ i}  \right)^i = O(1).
\end{align*}
By Markov's inequality this implies that $ Y_2 + Y_3 + Y_4 \le n^{ 1-  \sqrt{\delta }  } $ asymptotically almost surely.  Denote by $U$ the union of all $2$-, $3$- and $4$- cycles in $H$. Then the induced subhypergraph $H[V \setminus U ]$ has girth at least $5$ and, asymptotically almost surely, $|U| \le n^{1-\delta}$.

\end{proof}

We are now ready to prove Theorem~\ref{main_property}.
\begin{proof}[Proof of Theorem~\ref{main_property}] 
Our goal will be to find a subset $U \subset V$ of size $|U| \le n d^{ - \frac{1}{k-1 } } $ that (i) contains all cycles of length at most $4$ and every vertex of degree  more than $(1+\delta) d$; and (ii) such that,  every vertex $v $ in $ V \setminus U $ has at most $9 k^2 \left( \frac{ d }{ (\ln d)^2} \right)^{ \frac{1}{k-1} } = o \left(\left(  \frac{ d }{ \ln d}\right)^{ \frac{1}{k-1} }  \right)  $ neighbors in $U$.  Note that in this case, according to Lemma~\ref{simple_properties}, $H[U]$ is $k \left(\frac{ d  }{ (\ln d)^2 } \right)^{\frac{1}{k-1} } $-degenerate, concluding the proof assuming $d$ is sufficiently large. A similar idea has been used in~\cite{alon1997concentration,alon1999list,luczak1991chromatic}.

Towards that end, let $U_1$ be the set of vertices of degree more than $(1+\delta) d$, and $U_2 $ the set of vertices that are contained in a $2$-,$3$- or a $4$-cycle. Notice that $U_1, U_2,$ can be found in polynomial time and, according to Lemmas~\ref{part_b_random} and~\ref{part_c_random}, the size of $U_0:= | U_1 \cup U_2 |$ is at most $\frac{3n}{ d^2}$ for sufficiently large $n$ and $d$.

We now start with $U := U_0 $ and as long as there exists a vertex $v \in V \setminus U$ having at least $9 k^2 \left( \frac{ d }{ (\ln d)^2} \right)^{ \frac{1}{k-1} } $ neighbors in $U$ we do the following. Let $S_v = \{u_1, u_2, \ldots, u_{  N }  \}$ be the neighbors of $v$ in $U$. We choose an   arbitrary hyperedge   $h$  that contains $v$ and $u_1$ and update $U $ and $S_v$  by defining $U:= U \cup h $ and  $S_v := S_v \setminus  h $. We keep repeating this operation until $S_v$  is empty.

This process terminates with $|U | < n d^{- \frac{1}{ k-1}  } $ because, otherwise, we would get a subset $U \subset V$ of size $|U | =n d^{- \frac{1}{ k-1}  } $ spanning  more than   
\begin{align*}
 \frac{1}{k}  \left( \frac{ n}{ d^{ \frac{1}{k-1} }  }   - |U_0|    \right)  \times  9 k^2 \left( \frac{  d}{ (\ln d)^2 } \right)^{ \frac{1}{k-1} }   \times \frac{1}{k } >  \frac{ n}{d^{ \frac{1}{k-1} }  } \times \left( \frac{d}{(\ln d)^2 }  \right)^{ \frac{1}{k-1} } 
\end{align*}
 hyperedges, for sufficiently large $d$. According to Lemma~\ref{simple_properties} however,  $H$ does not contain any such  set asymptotically almost surely.

\end{proof}

\section{Acknowledgements}

The author is grateful to Dimitris Achlioptas, Irit Dinur and anonymous reviewers for detailed comments and feedback.

\bibliographystyle{plain}   

\bibliography{kolmo}

\end{document}